\documentclass[runningheads]{llncs}

\usepackage[T1]{fontenc}

\usepackage{amsmath}
\usepackage{graphicx}
\usepackage{amsthm}
\usepackage{amssymb}
\usepackage{ wasysym }
\usepackage[colorlinks=true, allcolors=blue]{hyperref}

\usepackage{siunitx}
\usepackage{listings}
\usepackage[dvipsnames]{xcolor}
\usepackage{graphicx}
\usepackage{pifont}
\newcommand{\cmark}{\ding{51}}%
\newcommand{\xmark}{\ding{55}}%
\usepackage{stmaryrd}
\usepackage{orcidlink}
\usepackage{mathtools}
\usepackage{subcaption}
\usepackage{cite}
\usepackage[frozencache,cachedir=.]{minted}
\usepackage{inconsolata} 
\usepackage[colorinlistoftodos]{todonotes}
\usepackage{tikz, tikz-cd}
\usepackage{booktabs, siunitx, multirow, caption, array}
\usepackage{longtable}
\renewcommand{\arraystretch}{1.2}
\newcolumntype{N}{>{\centering\arraybackslash}m{2.2cm}}
\newcolumntype{M}{>{\centering\arraybackslash}m{1.8cm}}

\usetikzlibrary{automata, arrows.meta, positioning, calc, decorations.pathreplacing, backgrounds, matrix, arrows, patterns, patterns.meta}

\def\scatterplotheight{2.8cm}

\usetikzlibrary{
  automata,
  positioning,
  decorations.pathmorphing,
  calc,
  shadings, 
 fadings,
  decorations.markings
}

\pgfdeclarehorizontalshading{linefade}{6cm}{%
  color(0cm)=(black);
  color(2cm)=(black);
  color(3.2cm)=(black);
  color(3.5cm)=(black);
  color(3.8cm)=(white);
  color(6cm)=(white)
}

\pgfdeclarehorizontalshading{arrowfade}{2cm}{%
  color(0cm)=(black);
  color(0.7cm)=(black);
  color(0.8cm)=(white);
  color(2cm)=(white)
}

\usetikzlibrary{decorations,decorations.markings} 
\pgfkeys{/tikz/.cd,
    contour distance/.store in=\ContourDistance,
    contour distance=-10pt, 
    contour step/.store in=\ContourStep,
    contour step=1pt,
}

\pgfdeclaredecoration{closed contour}{initial}
{%
\state{initial}[width=\ContourStep,next state=cont] {
    \pgfmoveto{\pgfpoint{\ContourStep}{\ContourDistance}}
    \pgfcoordinate{first}{\pgfpoint{\ContourStep}{\ContourDistance}}
    \pgfpathlineto{\pgfpoint{0.3\pgflinewidth}{\ContourDistance}}
    \pgfcoordinate{lastup}{\pgfpoint{1pt}{\ContourDistance}}
    
  }
  \state{cont}[width=\ContourStep]{
     \pgfmoveto{\pgfpointanchor{lastup}{center}}
     \pgfpathlineto{\pgfpoint{\ContourStep}{\ContourDistance}}
     \pgfcoordinate{lastup}{\pgfpoint{\ContourStep}{\ContourDistance}}
  }
  \state{final}[width=\ContourStep]
  { 
    \pgfmoveto{\pgfpointanchor{lastup}{center}}
    \pgfpathlineto{\pgfpointanchor{first}{center}}
  }
}

\newcommand{\paths}[1]{\text{Paths}(#1)}

\makeatletter
\newsavebox{\@brx}
\newcommand{\llangle}[1][]{\savebox{\@brx}{\(\m@th{#1\langle}\)}%
  \mathopen{\copy\@brx\mkern2mu\kern-0.9\wd\@brx\usebox{\@brx}}}
\newcommand{\rrangle}[1][]{\savebox{\@brx}{\(\m@th{#1\rangle}\)}%
  \mathclose{\copy\@brx\mkern2mu\kern-0.9\wd\@brx\usebox{\@brx}}}

  \newcommand{\nxt}{\bigcirc}
    \newcommand{\until}{\, U \,}

  \newcommand{\ltl}[0]{\text{LTL}}
  \newcommand{\ltln}[0]{\text{LTL}_{\setminus \nxt }}
  \newcommand{\ctln}[0]{\text{CTL}^{*}_{\setminus \nxt }}
  \newcommand{\ctlnn}[0]{\text{CTL}_{\setminus \nxt }}
\newcommand{\ctls}[0]{\text{CTL}^{*}}

  \definecolor{lpurple}{RGB}{251,181,255}
  \definecolor{lred}{RGB}{255,170,170}
  \definecolor{lcyan}{RGB}{223,255,252}
  \definecolor{black_opac}{RGB}{0,0,0}
  \tikzset{%
  prefix node name/.code={%
    \tikzset{%
      name/.code={\edef\tikz@fig@name{#1 ##1}}
    }%
  }%
}

\definecolor{lightred}{RGB}{241, 225, 222}
\definecolor{color1}{rgb}{0.1,0.498039215686275,0.9549019607843137}
\definecolor{alizarin}{rgb}{0.82, 0.1, 0.26}
\definecolor{antiquewhite}{rgb}{0.98, 0.92, 0.84}
\definecolor{azure}{rgb}{0.94, 1.0, 1.0}
\definecolor{offwhite}{rgb}{0.98, 0.95, 0.95}
\definecolor{pigment}{rgb}{0.2, 0.2, 0.6}

\newcommand{\stsp}[0]{S}
\newcommand{\init}[0]{I}
\newcommand{\trans}[0]{\rightarrow}
\newcommand{\obs}[1]{\llangle #1 \rrangle}

\newcommand{\quot}[0]{\mathcal{M}\hspace{-0.2pt}/\hspace{-0.2pt}_{\simeq}}
\newcommand{\qstsp}[0]{S\hspace{-1pt}/\hspace{-1pt}_{\simeq}}
\newcommand{\qinit}[0]{I\hspace{-1pt}/\hspace{-1pt}_{\simeq}}
\newcommand{\qtrans}[0]{\rightarrow\hspace{-3.5pt}/\hspace{-1.5pt}_\simeq\hspace{1pt}}
\newcommand{\qobs}[1]{\llangle #1 \rrangle\hspace{-1pt}/_\simeq}

\def\techreport{}

\begin{document}
\title{Branching Bisimulation Learning}
\titlerunning{Branching Bisimulation Learning}

\author{Alessandro Abate\inst{1}
\and
Mirco Giacobbe\inst{2}
\and\\
Christian Micheletti\inst{3}
\and
Yannik Schnitzer\inst{1}
}
\authorrunning{A. Abate et al.}
%
\institute{University of Oxford, Oxford, UK \\
\email{\{alessandro.abate,yannik.schnitzer\}@cs.ox.ac.uk}\\
 \and
University of Birmingham, Birmingham, UK \\
\email{m.giacobbe@bham.ac.uk} \and
University of Padua, Padua, Italy\\
\email{christian.micheletti@studenti.unipd.it}
}

\maketitle              
\begin{abstract}
We introduce a bisimulation learning algorithm for non-de\-ter\-mi\-nis\-tic transition systems. 
We generalise bisimulation learning to systems with bounded branching and extend its applicability to model checking branching-time temporal logic,
while previously it was limited to deterministic systems and model checking linear-time properties. 
Our method computes a finite stutter-insensitive bisimulation quotient of the system under analysis, represented as a decision tree. 
We adapt the proof rule for well-founded bisimulations to an iterative procedure that trains candidate decision trees from sample transitions of the system, and checks their validity over the entire transition relation using SMT solving. 
This results in a new technology for model checking CTL* without the next-time operator. 
Our technique is sound, entirely automated, and yields abstractions that are succinct and effective for formal verification and system diagnostics.   
We demonstrate the efficacy of our method on diverse benchmarks comprising concurrent software, communication protocols and robotic scenarios. 
Our method performs comparably to mature tools in the special case of LTL model checking, and outperforms the state of the art 
in CTL and CTL* model checking for systems with very large and countably infinite state space.
\end{abstract}
\section{Introduction}
Bisimulation establishes equivalence between transition systems, ensuring they exhibit identical observable behaviour across all possible computation trees. It captures not only the linear traces of system interactions but also the branching structure of their potential evolutions, making it a powerful characterisation for the abstraction of systems. 
This enables the efficient analysis of complex concrete systems by reducing them to simpler abstract systems 
that capture the behaviour that is essential for model checking a formal specification.
When the concrete system and its abstract counterpart are in a bisimulation relation, their observable branching behaviour is indistinguishable. As a result, model checking linear- and branching-time temporal logic yield the same result on the abstract and the concrete system, even in the presence of non-determinism.

Algorithms for computing bisimulations have been developed extensively for systems with finite state spaces. The Paige-Tarjan algorithm for partition refinement laid the foundations for the automated constructions of bisimulation relations and their respective quotients~\cite{DBLP:journals/siamcomp/PaigeT87}. Partition refinement is the state of the art for this purpose, and has lent itself to extension towards on-the-fly quotient construction and symbolic as well as parallel implementations~\cite{DBLP:conf/stoc/LeeY92,DBLP:conf/cav/BouajjaniFH90,DBLP:conf/cav/LeeR94,DBLP:journals/sttt/DijkP18,DBLP:conf/facs2/0001GHHW21}. 
These produce the coarsest bisimulation quotient, namely the most succinct representation possible of a bisimulation. Yet, this  
can be prohibitively expensive to compute for systems with large state spaces and complex arithmetic pathways, both when using explicit-state and symbolic algorithms. Consequently, model-checking techniques based on the partition refinement algorithm are typically restricted to systems with smaller state spaces and simpler transition logic.

Counterexample-guided abstraction refinement (CEGAR) has enabled the incremental construction 
of abstract systems using satisfiability modulo theory (SMT) solvers, benefitting from their ability to reason effectively 
over arithmetic constraints~\cite{DBLP:conf/cav/ClarkeGJLV00,DBLP:series/faia/BarrettSST21}. Modern software, hardware and cyber-physical systems are rich in arithmetic operations and have often very large state spaces. As a result, significant progress in software model checking and the verification of cyber-physical systems has been driven by advancements in CEGAR~\cite{DBLP:conf/popl/HenzingerJMS02,DBLP:journals/sttt/BeyerHJM07,DBLP:conf/cav/BeyerK11,tacas14:kojak,theta-fmcad2017,DBLP:conf/cav/KahsaiRSS16}, 
specifically designed to compute simulation quotients to prove safety properties~\cite{DBLP:conf/cav/GrafS97}. 
However, for model checking liveness properties and linear-time logic, CEGAR requires non-trivial adaptation~\cite{DBLP:conf/sas/CookPR05,DBLP:conf/popl/CookGPRV07,DBLP:conf/cav/PrabhakarS16}, 
and branching-time logic is entirely out of scope for simulation quotients.

Bisimulation learning has introduced an incremental approach to computing bisimulation quotients~\cite{DBLP:conf/cav/AbateGS24}, where 
information about the system is learned from counterexample models---as happens in counterexample-guided inductive synthesis (CEGIS)~\cite{DBLP:conf/asplos/Solar-LezamaTBSS06}---as opposed to counterexample proofs---as happens in CEGAR~\cite{DBLP:conf/popl/HenzingerJMM04,DBLP:conf/cav/McMillan06}. 
Bisimulation learning relies on a parameterised representation of a quotient, such as a decision tree, 
and trains its parameters to {\em fit} a bisimulation quotient over sampled transitions of the system. 
Then, an adversarial component, such as an SMT solver, is used to check whether the quotient satisfies the bisimulation property 
across the entire state space and, upon a negative answer, propose counterexample transitions for further re-training of the quotient parameters. 
While in principle bisimulation relations preserve branching-time logic and support non-determinism, 
bisimulation learning was limited to deterministic systems and linear-time logic specifications~\cite{DBLP:conf/cav/AbateGS24,DBLP:conf/birthday/AbateGRS25}. 

We introduce a new, generalised bisimulation learning algorithm. 
Our result builds upon the proof rule for well-founded bisimulation~\cite{DBLP:conf/fsttcs/Namjoshi97}, 
which we specialise to systems  with bounded branching and 
integrate it within a bisimulation learning algorithm for non-deterministic systems.
While well-founded bisimulations were originally developed for interactive theorem proving, 
our result enables their fully automated construction.
We represent finite stutter-insensitive bisimulation quotients 
of non-deterministic systems, and enable the effective abstraction of systems with very large or countably infinite state space. 

Stutter-insensitive bisimulations are indistinguishable to an external observer that cannot track intermediate
transitions lacking an observable state change in a system. 
It is a standard result that stutter-insensitive bisimulation is an abstract semantics 
for $\ctln$---the branching-time logic $\ctls$ without the next-time operator~\cite{DBLP:journals/jacm/NicolaV95}. 
In other words, when the abstract system corresponds to a stutter-insensitive bisimulation quotient and distinguishes at least the atomic 
propositions of the $\ctln$ formula $\phi$,
then the answers to model checking $\phi$ on the abstract system and the concrete system agree. 

We demonstrate the efficacy of our method on a standard set of benchmarks for the formal verification of finite and 
infinite state systems against linear-time and branching-time temporal logic. 
We compare the runtime performance of our prototype with mature tools for the termination analysis of 
software (CPAChecker and Ultimate), LTL and CTL model checking finite-state systems (nuXmv), 
LTL model checking infinite-state systems (nuXmv and UltimateLTL), and 
CTL* model checking infinite-state systems (T2). 
Our method performs comparably to mature tools for termination analysis, finite-state model checking, 
and LTL model checking infinite-state systems, while yielding 
superior results in branching-time model checking infinite-state systems. 
Overall, our technology addresses the general CTL* model checking problem for infinite-state systems,
performs comparably to specialised tools for linear-time properties, 
and establishes a new state of the art in branching-time model checking. 

Our contribution is threefold. 
First, we generalise bisimulation learning to non-de\-ter\-mi\-nis\-tic transition systems, 
leveraging the proof rule for well-founded bisimulations which we fully automate.  
Second, we introduce a new model checking algorithm for infinite-state systems against CTL* without the next-time operator.
Third, we demonstrate the efficacy of our general approach on standard benchmarks which not only compares favourably 
with the state of the art but also provides more informative results. 
Our method produces succinct stutter-insensitive bisimulation quotients of the system under analysis. In conjunction with a fix-point based model checking algorithm, this yields the exact initial conditions for which the concrete system satisfies a given specification. 

\section{Model Checking}
\label{sec:modelchecking}
We consider the problem of determining whether state transition systems with countable (possibly infinite) state space satisfy 
linear-time and branching-time temporal logic specifications. More precisely, we consider the model checking problem of non-deterministic transition systems with bounded branching with respect to linear-time and branching-time specifications expressed in $\ctls$ without the next-time operator.
This encompasses a broad variety of formal verification questions for software systems
with concurrency, reactive systems including synchronisation and communication protocols, 
as well as cyber-physical systems over finite or infinite discrete grid worlds.

\begin{definition}[Transition Systems]
    A transition system ${\cal M}$ consists of
    \begin{itemize}
	   \item a state space $\stsp$,
 	 \item an initial region $\init \subseteq \stsp$, and
	   \item a transition relation $\trans \subseteq \stsp \times \stsp$.
\end{itemize}
We consider transition systems that are non-blocking and have bounded branching. 
In other words, every state $s \in S$ has at least one and at most $k \in \bbbn$ successors, i.e., 
$0 < |\{ t \in S \colon s \trans t \}| \leq k$, where we say that $\cal M$ has $k$-bounded branching. 
We say that $\cal M$ is labelled when it additionally comprises
\begin{itemize}
    \item a set of atomic propositions $\Pi$ (the observables), and
    \item a labelling (or observation) function $\llangle \cdot \rrangle \colon S \to {\cal P}(\Pi)$.
\end{itemize}
A path $\pi = s_0s_1\dots$ of $\cal M$ is a sequence of states such that $s_i \trans s_{i+1}$, for all $i \geq 0$. 
Since transition systems are non-blocking, every path extends to infinite length. We denote the set of all infinite paths starting in state $s$ as $\text{Paths}(s)$.
    \label{def:ts}
\end{definition}

\begin{remark}[Bounded vs. Finite Branching] Bounded branching requires a constant $k$ that bounds the number of successors across the entire state space, whereas \emph{finite branching} is weaker, requiring every state to have finitely many successors, not imposing a constant upper bound. We note that bounded branching is a mild modelling restriction, which encompasses concurrency with finitely many processes, 
non-deterministic guarded commands and conditional choices, and non-deterministic variables or inputs with finite static domain. 
It excludes non-deterministic variables with infinite or state-dependent domain size.
\qed 
\end{remark}

We consider the branching-time temporal logic $\text{CTL}^*$ without the next-time operator ($\ctln$) 
as our formal specification language for the temporal behaviour of systems~\cite{DBLP:books/daglib/baierkatoen,DBLP:conf/popl/EmersonH83}. 
This subsumes and generalises 
Linear Temporal Logic (LTL)~\cite{DBLP:conf/focs/Pnueli77} and Computation Tree Logic~(CTL)~\cite{DBLP:conf/lop/ClarkeE81} (excluding the next-time operator), expressing both linear-time and branching-time properties. The $\ctln$ formulae are 
constructed according to the following grammar:
\begin{gather*}
    \phi ::= \text{true} \mid p \in \Pi \mid \phi \wedge \phi \mid \neg \phi \mid \exists \psi\\
    \psi ::= \phi \mid \psi \wedge \psi \mid \neg \psi \mid \psi \until \psi.
\end{gather*}
The model checking problem for $\ctln$ is to decide whether transition system $\cal M$ satisfies a given $\ctln$ formula $\phi$. The satisfaction relation $\models$ for state formulae $\phi$ is defined over states $s\in S$ by: 
\begin{alignat*}{4}
	s & \models \text{true}\\
	s &\models p &&\text{iff }&& p \in \obs{s} \\
	s &\models \phi_1 \wedge \phi_2 \quad &&\text{iff } &&s \models \phi_1 \text{ and } s \models \phi_2\\
	s &\models \neg \phi \quad &&\text{iff } &&s \not\models \phi\\
    s &\models \exists \psi \quad &&\text{iff } && \exists \pi \in \paths{s} \colon \pi \models \psi.
\end{alignat*}
For a path $\pi$, the satisfaction relation $\models$ for path formulae $\psi$ is given by:
\begin{alignat*}{4}
	  \pi &\models \phi &&\text{iff }&& s_0 \models \phi \\
	\pi &\models \psi_1 \wedge \psi_2 \quad &&\text{iff } &&\pi \models \psi_1 \text{ and } \pi \models \psi_2\\
	\pi &\models \neg \psi \quad &&\text{iff } &&\pi \not\models \psi\\
    \pi &\models \psi_1 \until \psi_2 &&\text{iff } &&
    \exists k \in \mathbb{N} \colon \pi[k..] \models \psi_2 \text{ and}  \\
    &&&&&\forall 0 \leq l < k \colon \pi[l..] \models \psi_1,
\end{alignat*}
where for path $\pi = s_0s_1\dots$ and $i \in \mathbb{N}$, $\pi[i..] = s_is_{i+1}\dots$ denotes the suffix starting from index $i$. The satisfaction relation of a state formula $\phi$ is lifted to the entire transition system by requiring that every initial state must satisfy $\phi$:
$$ \mathcal{M} \models \phi \text{ iff } \forall s \in I \colon s \models \phi.$$
We also introduce the derived path operators \textit{"eventually"} $\lozenge$ and \textit{"globally"} $\square$. The formula $\lozenge \psi := \text{true} \until \psi$ states that $\psi$ must be true in some state on the path. The formula $\square \psi := \neg(\lozenge \neg \psi)$ requires that $\psi$ holds true in all states of the path. The universal quantification over paths $\forall \psi$ can be expressed as $\neg \exists \neg \psi$.
We do not include the \textit{"next-time"} operator $\nxt$ from full CTL\textsuperscript{*}, since we are interested in stutter-insensitive bisimulations, which do not preserve a system's stepwise behaviour, as expressed by $\nxt$.

\section{Stutter-Insensitive Bisimulations}

This section introduces the concept of abstraction, specifically that of stutter-insensitive bisimulation, which preserves all linear- and branching-time behaviour up to externally unobservable stutter steps, as captured by $\ctln$. 

\begin{definition}[Partitions]
     A partition on $\cal M$ is an equivalence relation $\simeq \subseteq S \times S$ on $S$,
     which defines the quotient space $S/_{\simeq}$ of pairwise-disjoint regions of $S$ whose union is $S$, i.e., 
     $\qstsp$ is the set of equivalence classes of $\simeq$. 
     A partition is called label-preserving (or observation-preserving) iff $s \simeq t \implies \obs{s} = \obs{t}$.
\end{definition}
A partition induces an abstract transition system -- the \emph{quotient}, which aggregates equivalent states and their behaviours into representative states.

\begin{definition}[Quotient]
     The quotient of $\cal M$ under the partition $\simeq$ is the transition system ${\cal M}/_\simeq$ with
\begin{itemize}  
    \item state space $\qstsp$,
    \item initial region $\qinit$ where $R \in \qinit$ iff $R \cap I \neq \emptyset$, and
\item transition relation $\qtrans$, with $R \qtrans Q$ iff either: 
\begin{enumerate}
    \item $R\neq Q$ and $\exists s \in R, t \in Q \colon s \rightarrow t$,
    \item $R = Q$ and $\forall s \in R \, \exists t \in R \colon s \rightarrow t$.
\end{enumerate}
\end{itemize}
If the partition $\simeq$ is label-preserving, the quotient further comprises a well-defined labelling function given by $\qobs{R} = \obs{s}$, for any $s\in R$.
\label{def:quotient}
\end{definition}

An \emph{abstract} state in the quotient represents an equivalence class of the underlying partition, inheriting the behaviours of included states. Depending on the partition, the quotient preserves temporal properties of the concrete system, such that model checking results carry over from the  quotient~\cite{DBLP:journals/siamcomp/PaigeT87}.

A prominent notion of equivalence on states is \emph{strong bisimilarity}~\cite{DBLP:journals/tcs/BrowneCG88,DBLP:journals/jacm/HennessyM85}. Strong bisimilarity demands that related states can replicate each other's transitions, with transitions leading to states that are themselves related. This ensures temporal equivalence in both linear- and branching-time, making it a suitable abstract semantics for full $\ctls$, including the \emph{next}-operator~\cite{DBLP:books/sp/Milner80}. However, preserving exact stepwise behavior limits the potential for state-space reduction when constructing the corresponding quotient. Alternatively, we consider the weaker notion of \emph{stutter-insensitivity}, which abstracts from exact steps in the concrete system that are externally unobservable. This results in potentially smaller quotients while preserving specifications expressible in $\ctln$.

\begin{definition}[Stutter-insensitive Bisimulation]
A label-preserving partition $\simeq$ is a stutter-insensitive bisimulation if, for all states $s, s' \in S$ with $s \simeq s'$ and paths $\pi\in \paths{s}$, there exists a path $\pi' \in \paths{s'}$, such that $\pi$ and $\pi'$ can be split into an equal number of non-empty finite subsequences $\pi = B_1B_2\dots$ and $\pi' = B_1'B_2'\dots$, for which it holds that $\forall i \geq 0 \, \forall t \in B_i, t' \in B_i' \colon t \simeq t'$.
\label{def:stuttBisim}
\end{definition}

Stutter-insensitive bisimulation requires that related states have outgoing paths composed of identical subsequences of equivalence classes, regardless of the lengths of these subsequences. This guarantees that related states are roots to computation trees that are externally indistinguishable up to exact step counts. Figure~\ref{fig:stuttbisim} illustrates this condition.

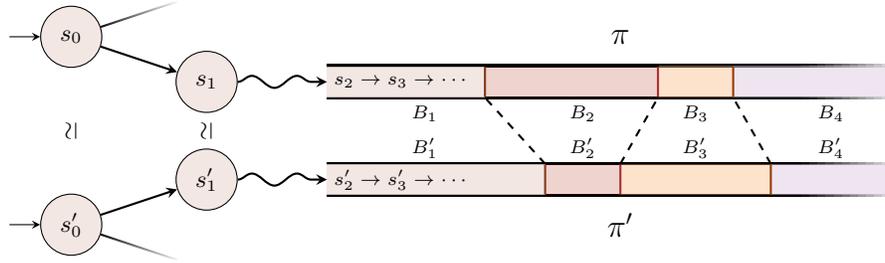
\begin{figure}[t]
    \centering
\begin{tikzpicture}[>=stealth, auto, on grid, node distance=1.8cm]

\node[state, initial, initial text=, fill = lightred, fill opacity = 0.8, text opacity = 1] (A) at (0,0) {$s_0'$};
\node[state] (B) [right=of A, yshift=0.6cm, fill = lightred, fill opacity = 0.8, text opacity = 1] {$s_1'$};
\node[state, draw=none] (C) [right=of A, yshift=-0.6cm] {};

\draw[->, thick] (A) -- (B);
\draw[-, thick, path fading=south, draw=black] 
      (A) -- (C);

\coordinate (RectStart) at ($(B) + (1.5,0)$);
\coordinate (RectEnd)   at ($(RectStart) + (6.5,0)$);

\shade[shading=linefade, draw=none]
  (RectStart) + (0.1,  0.23) rectangle ++(8,  0.2);
\shade[shading=linefade, draw=none]
  (RectStart) + (0.1, -0.23) rectangle ++(8, -0.2);

   \fill[fill = lightred, fill opacity = 0.8]
    ($(RectStart) + (0.1,  -0.2)$)
    rectangle
    ($(RectStart) + (3,  0.2)$);
   \fill[BrickRed, opacity=0.21]
    ($(RectStart) + (3.01,  -0.2)$)
    rectangle
    ($(RectStart) + (4,  0.2)$);    
   \fill[Apricot, opacity=0.4]
    ($(RectStart) + (4.01,  -0.2)$)
    rectangle
    ($(RectStart) + (6,  0.2)$);    
\fill[Plum!20, opacity=0.5]
    ($(RectStart) + (6.01, -0.2)$)
    rectangle
    ($(RectStart) + (7.55, 0.2)$);

  \draw[thick, Sepia!80]
    ($(RectStart) + (3,  -0.2)$) -- ($(RectStart) + (3,  0.2)$);
  \draw[thick, Maroon]
    ($(RectStart) + (4.0,  -0.2)$) -- ($(RectStart) + (4.0,  0.2)$);  
  \draw[thick, RawSienna]
    ($(RectStart) + (6,  -0.2)$) -- ($(RectStart) + (6,  0.2)$);

    \coordinate (Purple1) at  ($(RectStart) + (3,  0.21)$);
    \coordinate (Blue1) at  ($(RectStart) + (4,  0.21)$);
    \coordinate (Plum1) at  ($(RectStart) + (6,  0.21)$);

\draw[-stealth, thick,
      decorate,
      decoration={snake, amplitude=0.8mm, segment length=6mm}]
      (B) -- ($(RectStart) + (0.1,0)$);

      \node[] (pi1) at ($(RectStart) + (4.0,-0.6)$) {\large$\pi'$};

      \node[] (block1) at ($(RectStart) + (1.4,0.42)$) {\scriptsize$B_1'$};
    \node[] (block2) at ($(RectStart) + (3.5,0.42)$) {\scriptsize$B_2'$};
    \node[] (block3) at ($(RectStart) + (5,0.42)$) {\scriptsize$B_3'$};
    \node[] (block4) at ($(RectStart) + (6.8,0.42)$) {\scriptsize$B_4'$};

\begin{scope}[yshift=2.5cm]
  \node[state, initial, initial text=,fill = lightred, fill opacity = 0.8, text opacity = 1] (A2) at (0,0) {$s_0$};
  \node[state, fill = lightred, fill opacity = 0.8, text opacity = 1] (B2) [right=of A2, yshift=-0.6cm] {$s_1$};
  \node[state, draw=none] (C2) [right=of A2, yshift=0.6cm] {};

  \draw[->, thick] (A2) -- (B2);
  \draw[-, thick, path fading=north, draw=black] 
        (A2) -- (C2);

  \coordinate (RectStart2) at ($(B2) + (1.5,0)$);
  \coordinate (RectEnd2)   at ($(RectStart2) + (6.5,0)$);

  \shade[shading=linefade, draw=none]
    (RectStart2) + (0.1,  0.23) rectangle ++(8,  0.2);
  \shade[shading=linefade, draw=none]
    (RectStart2) + (0.1, -0.23) rectangle ++(8, -0.2);


   \fill[fill = lightred, fill opacity = 0.8]
    ($(RectStart2) + (0.1,  -0.2)$)
    rectangle
    ($(RectStart2) + (2.2,  0.2)$);
   \fill[BrickRed, opacity=0.21]
    ($(RectStart2) + (2.21,  -0.2)$)
    rectangle
    ($(RectStart2) + (4.5,  0.2)$);    
   \fill[Apricot, opacity=0.4]
    ($(RectStart2) + (4.51,  -0.21)$)
    rectangle
    ($(RectStart2) + (5.5,  0.21)$);    
\fill[Plum!20, opacity=0.5]
    ($(RectStart2) + (5.5, -0.21)$)
    rectangle
    ($(RectStart2) + (7.55, 0.21)$);


  \draw[thick, Maroon]
    ($(RectStart2) + (4.5,  -0.2)$) -- ($(RectStart2) + (4.5,  0.2)$);  
  \draw[thick, Sepia!80]
    ($(RectStart2) + (2.2,  -0.2)$) -- ($(RectStart2) + (2.2,  0.2)$);
  \draw[thick, RawSienna]
    ($(RectStart2) + (5.5,  -0.2)$) -- ($(RectStart2) + (5.5,  0.2)$);
    \coordinate (Purple2) at ($(RectStart2) + (2.2,  -0.21)$);
    \coordinate (Blue2) at  ($(RectStart2) + (4.5,  -0.21)$);
    \coordinate (Plum2) at  ($(RectStart2) + (5.5,  -0.21)$);

\draw[-stealth, thick,
      decorate,
      decoration={snake, amplitude=0.8mm, segment length=6mm}]
      (B2) -- ($(RectStart2) + (0.1,0)$);

      \node[] (pi2) at ($(RectStart2) + (4.0,0.6)$) {\large$\pi$};

      \node[] (block1) at ($(RectStart2) + (1.4,-0.42)$) {\scriptsize$B_1$};
    \node[] (block2) at ($(RectStart2) + (3.5,-0.42)$) {\scriptsize$B_2$};
    \node[] (block3) at ($(RectStart2) + (5,-0.42)$) {\scriptsize$B_3$};
    \node[] (block4) at ($(RectStart2) + (6.8,-0.42)$) {\scriptsize$B_4$};
\end{scope}

\draw[dashed, thick] (Purple1) -- (Purple2);
\draw[dashed, thick] (Blue1) -- (Blue2);
\draw[dashed, thick] (Plum1) -- (Plum2);

\node[] (states1) at ($(RectStart) + (1.1,0)$) {\scriptsize$s_2' \rightarrow s_3'\rightarrow \cdots$};
\node[] (states2) at ($(RectStart2) + (1.1,0)$) {\scriptsize$s_2 \rightarrow s_3\rightarrow \cdots$};

\node[] (sim1) at ($(A2) + (0,-1.25)$) {\rotatebox{90}{$\simeq$}};
\node[] (sim2) at ($(B2) + (0,-0.65)$) {\rotatebox{90}{$\simeq$}};

\end{tikzpicture}

    \caption{States related under stutter-insensitive bisimulation, with examples of matching paths, consisting of identical subsequences of equivalence classes.}

    \label{fig:stuttbisim}
\end{figure}

\begin{remark}
The case distinction for the quotient transition relation in Definition~\ref{def:quotient} is needed to prevent spurious self-loops in the quotient system, which may arise from unobservable stuttering steps in the original system~\cite{DBLP:conf/fsttcs/Namjoshi97}. We have to exclude unobservable intra-class transitions from the quotient, except when they reflect diverging behaviour, i.e., when all states within a class can remain in the class indefinitely (Case 2). Figure~\ref{fig:bisimquotientex} shows an example system and its stutter-insensitive bisimulation quotient. \qed 
\end{remark}
\tikzset{every loop/.style={-stealth, thick}}
\tikzset{
  every initial by arrow/.style={-stealth, thick}
}
\begin{figure}[ht]

\centering
\hspace*{4em}
\scalebox{0.99}{
\begin{tikzpicture} [node distance = 1.8cm, on grid, auto]
\begin{scope}[name prefix =G1, xshift= 0cm]
    \node (q0) [state, initial,initial text = {},initial above, fill = lightred, fill opacity = 0.8, text opacity = 1] {$s_0$};
    \node (q1) [state,,initial text = {},below left = of q0, fill = lightred, fill opacity = 0.8, text opacity = 1] {$s_1$};
    \node (q2) [state, below right = of q0, fill = BrickRed, fill opacity=0.35,, text opacity = 1] {$s_2$};
    \node (q3) [state, below right = of q1,fill = BrickRed, fill opacity=0.35,, text opacity = 1] {$s_3$};
    \node (q4) [state, below = of q3,fill = Apricot, fill opacity = 0.5, text opacity = 1, yshift = 0.2cm] {$s_4$};
       
    \node [left of=q0, xshift=1cm] {$\{a\}$};
    \node [left of=q1, xshift=1cm] {$\{a\}$};
    \node [right of=q2, xshift=-1cm] {$\{b\}$};
    \node [left of=q3, xshift=1cm] {$\{b\}$};
    \node [left of=q4, xshift=1cm] {$\{c\}$};
       
    \node [below of=q4, yshift = 1cm, xshift = 1.55cm] {\parbox{0.3\linewidth}{{\large{$\cal M$}}\label{subfig:a}}};
       
    \path [-stealth, thick]
        (q0) edge[]  (q1)
        (q1) edge [] (q3)
        (q3) edge [bend left, shorten >= 4mm, out =40] (q2)
        (q2) edge [bend left, shorten >= 4mm, out = 40] (q3)
        (q2) edge [] (q0)
        (q3) edge [] (q4)
        (q4) edge [loop right, -stealth, thick]  node {$ $}();
\end{scope}

\begin{scope}[xshift = 5cm, name prefix =G2, yshift = -0.4cm]
    \node (q0) [state, initial,initial text = {},initial above, fill = lightred, fill opacity = 0.8, text opacity = 1] {$R$};
    \node (q2) [state, below = of q0, fill = BrickRed, fill opacity=0.35,, text opacity = 1] {$Q$};
    \node (q4) [state, below = of q2,fill = Apricot, fill opacity = 0.5, text opacity = 1, yshift = 0cm] {$P$};
       
    \node [left of=q0, xshift=1cm] {$\{a\}$};
    \node [left of=q2, xshift=1cm] {$\{b\}$};
    \node [left of=q4, xshift=1cm] {$\{c\}$};
       
    \node [below of=q4, yshift = 0.8cm, xshift = 1.55cm] {\parbox{0.3\linewidth}{{\large{${\cal M}/_\simeq $}}\label{subfig:a2}}};
       
    \path [-stealth, thick]
        (q0) edge [bend left, shorten >= 4mm, out = 40] (q2)
        (q2) edge [bend left, shorten >= 4mm, out = 40] (q0)
        (q2) edge [] (q4)
        (q4) edge [loop right, -stealth, thick]  node {$ $}()
        (q2) edge [loop right, -stealth, thick]  node {$ $}();
\end{scope}
\end{tikzpicture}
}

\caption{A system $\cal M$ and a possible stutter-insensitive bisimulation quotient ${\cal M}/_\simeq $. The abstract state $R$ lacks a self-loop despite an intra-class transition, as all represented concrete states eventually transition to $Q$. Conversely, $Q$ has a self-loop since its concrete states can remain within the class indefinitely.}
\label{fig:bisimquotientex}

\end{figure}
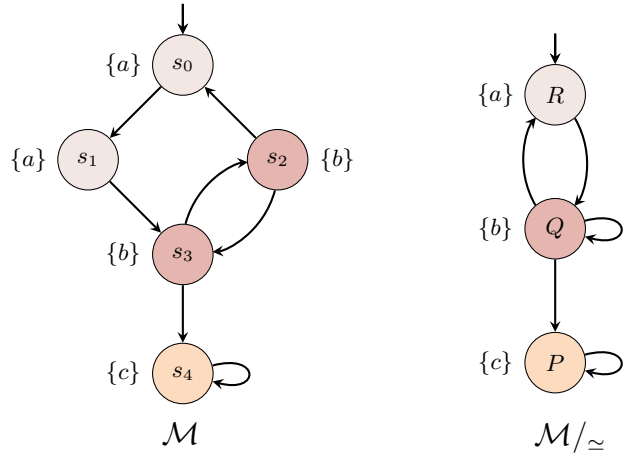

Since states related by a stutter-insensitive bisimulation are indistinguishable in linear- and branching-time up to exact stepwise behaviour, the model checking question for $\ctln$, as introduced in Section~\ref{sec:modelchecking}, yields the same result for a system and its potentially much smaller stutter-insensitive bisimulation quotient. 

\begin{theorem}[\hspace{-0.1pt}\protect{\cite[Theorem 4.2]{DBLP:journals/tcs/BrowneCG88}}]
    \label{thm:divsens}
    Let ${\cal M}$ be a labelled transition system and let $\simeq$ be a label-preserving stutter-insensitive bisimulation on ${\cal M}$. For every $\ctln$ formula $\phi$, it holds that ${\cal M} \models \phi$ if and only if ${\cal M}/_\simeq \models \phi$. \qed 
\end{theorem}

\section{Well-founded Bisimulations with Bounded Branching}

We revisit \emph{well-founded} bisimulation~\cite{DBLP:conf/fsttcs/Namjoshi97, DBLP:conf/cav/ManoliosNS99} as an alternative formulation of stutter-insensitive bisimulation. This notion builds on entirely local conditions for states and their transitions that ensure the existence of matching infinite paths. We show how this definition generalises the deterministic version used in bisimulation learning \cite{DBLP:conf/cav/AbateGS24} and characterises stutter-insensitive bisimulations for non-deterministic systems. By specialising this notion to systems with bounded branching, we are able to effectively encode the local conditions into quantifier-free first order logic formulas over states and their successors, which we leverage in the design of a counterexample-guided bisimulation learning approach. To the best of our knowledge, this is the first approach for the fully automatic computation of well-founded bisimulations on potentially infinite-state non-deterministic systems, which was originally designed as a proof rule for theorem proving and relied on manually crafted relations. If terminated successfully, the novel bisimulation learning procedure generates a finite quotient system, enabling state-of-the-art finite-state model checkers to verify $\ctln$ properties whose results directly carry over to the original, possibly infinite system.

The proof rule for well-founded bisimulations uses \emph{ranking functions} over a well-founded set, which decrease along stuttering steps and eventually lead to a matching transition. Since well-founded sets preclude infinite descending sequences, this guarantees the finiteness of stuttering and ensures the existence of matching infinite paths between related states based solely on local conditions.

\begin{theorem}[\!\!\protect{\cite[Theorem 1]{DBLP:conf/fsttcs/Namjoshi97}}]
\label{thm:namjoshi}
    Let ${\cal M}$ be a transition system and let $\simeq$ be a label-preserving partition on $\cal M$. Suppose there exists a function $r \colon S \times S \to \mathbb{N}$ such that, for every $s,s',t \in S$ with $s \simeq t$ and $s \trans s'$, the following holds:
    \begin{align}
        &\exists t' \in S \colon t \trans t' \land s' \simeq t' \; \lor \label{eq:first}\\
        &s \simeq s' \land r(s',s' )< r(s,s) \; \lor \label{eq:sec}\\
        &\exists t' \in S \colon t \trans t' \land t \simeq t' \land r(s',t') < r(s',t). \label{eq:third}
    \end{align}
    Then $\simeq$ is a stutter-insensitive bisimulation on $\mathcal{M}$. \qed
\end{theorem}

In this work, we focus on ranking functions $r \colon S \times S \to \mathbb{N}$ mapping to the natural numbers, though all results extend to arbitrary well-founded sets. 

The intuition behind Theorem~\ref{thm:namjoshi} is as follows: For any states $s \simeq t$ with $s \rightarrow s'$, there are three possibilities. First, there may be an immediate matching transition for $t$ (Case~\eqref{eq:first}). If the first case does not apply, and $s \simeq s'$, the rank decreases (Case~\eqref{eq:sec}). Since $r$ is defined over a well-founded set bounded from below, this can happen only finitely many times. In the remaining case, where $s \not\simeq s'$, there must exist a transition $t \rightarrow t'$ such that $t \simeq t'$ and the rank decreases (Case~\eqref{eq:third}). Again, by the well-foundedness of $(\mathbb{N}, <)$, this can also happen only finitely many times. 
Thus, a state related to $s'$ is eventually reached from $t$ after at most finite stuttering in the same equivalence class~\cite{DBLP:conf/fsttcs/Namjoshi97}. This ensures the existence of matching infinite paths, as per Definition~\ref{def:stuttBisim}, based solely on local reasoning over states and their immediate transitions.

The notion of well-founded bisimulations in Theorem~\ref{thm:namjoshi} generalises the proof rule used in deterministic bisimulation learning~\cite{DBLP:conf/cav/AbateGS24}, which essentially consists of restricted sub-conditions of Cases~\eqref{eq:first} and~\eqref{eq:sec} and is thus limited to deterministic systems~\cite[Theorem 2]{DBLP:conf/cav/AbateGS24}. Notice that, since states in non-deterministic systems can have multiple outgoing transitions, the proof rule in Theorem~\ref{thm:namjoshi} incorporates quantification over successors. For the design of an effective and efficient counterexample-guided bisimulation learning approach to the synthesis of well-founded bisimulations and their quotients, it is crucial to express these conditions in a quantifier-free fragment of a decidable first-order theory~\cite{DBLP:conf/cav/AbateDKKP18}. In the deterministic case addressed in previous work~\cite{DBLP:conf/cav/AbateGS24}, this is straightforward, since each state has a single outgoing transition, so the conditions do not involve quantification. To address this in the non-deterministic case, instead, we leverage the assumption of bounded branching, which allows the quantification over successors to be reformulated as a finite disjunction.

For a transition system $\cal M$ that has $k$-bounded branching, the transition relation $\rightarrow$ can be expressed as the union of $k$ deterministic transition functions $\sigma_i \colon S \to S$, for $1 \leq i \leq k$:
\begin{equation}
    \rightarrow ~=~ \bigcup_{i=1}^k \sigma_i.
\end{equation}
We assume that every state has exactly $k$ successors. This assumption is without loss of generality, as states with fewer than $k$ successors can have their successors duplicated. With that, we can state a version of Theorem~\ref{thm:namjoshi} for transition systems with bounded branching that eliminates quantification over successors.

\begin{theorem}
    \label{thm:boundednamjoshi}
    Let ${\cal M}$ be a transition system that has $k$-bounded branching and let $\simeq$ be a label-preserving partition on $\cal M$. Suppose there exists a function $r \colon S \times S \to \mathbb{N}$ such that, for every $s,t \in S$ with $s \simeq t$, the following holds:
    \begin{align}
        \bigwedge_{i = 1}^k \Big( &\bigvee_{j = 1}^k \sigma_i(s) \simeq \sigma_j(t) \; \lor \label{eq:qffirst}\\
        &s \simeq \sigma_i(s) \land r\left(\sigma_i(s),\sigma_i(s)\right)< r(s,s) \; \lor \label{eq:qfsec}\\
        &\bigvee_{j = 1}^k t \simeq \sigma_j(t)  \land r(\sigma_i(s),\sigma_j(t)) < r(\sigma_i(s),t)\Big). \label{eq:qfthird}
    \end{align}
    Then $\simeq$ is a stutter-insensitive bisimulation on $\mathcal{M}$. 
\end{theorem}
\begin{proof}
We show that Equation~\eqref{eq:qffirst} implies Equation~\eqref{eq:first} of Theorem~\ref{thm:namjoshi}. The remaining disjuncts can be treated analogously. Let \(s, t \in S\) such that $s \simeq t$ and
\begin{equation*}
    \bigwedge_{i = 1}^k \Big(\bigvee_{j = 1}^k \sigma_i(s) \simeq \sigma_j(t) \Big).
\end{equation*}
Since \(\mathcal{M}\) has \(k\)-bounded branching, both \(s\) and \(t\) have exactly \(k\) successors, represented by the \(k\) deterministic transition functions \(\sigma_i, 1 \leq i \leq k\). A conjunction over all \(k\) successors corresponds to a universal quantification over the successors, while a disjunction corresponds to an existential quantification. Consequently, the above expression can be rewritten as:
\begin{align*}
    &\forall s' \in S \colon \Big(s \to s' \implies \bigvee_{j = 1}^k s' \simeq \sigma_j(t)\Big)\\
    \Leftrightarrow \; & \forall s' \in S \colon \Big(s \to s' \implies \exists t' \in S \colon (t \to t' \land s' \simeq t')\Big),
\end{align*}
which is precisely Equation~\eqref{eq:first} of Theorem~\ref{thm:namjoshi}.
\end{proof}

\section{Bisimulation Learning for Non-Deterministic Systems}

We can now leverage the quantifier-free formulation of well-founded bisimulations from Theorem~\ref{thm:boundednamjoshi} to design a counterexample-guided bisimulation learning procedure for synthesising stutter-insensitive bisimulation quotients. The problem of identifying a suitable partition and ranking function that satisfy the conditions of a well-founded bisimulation is framed as a learning problem. To this end, we introduce the concept of \emph{state classifiers}.

\begin{definition}[State Classifier]
    A state classifier on a labelled transition system with state space $S$ is any function $f \colon S \to C$ that maps states to a finite set of classes $C$. It is label-preserving if $f(s) = f(t)$ implies 
    $\llangle s \rrangle = \llangle t \rrangle$. A classifier induces the partition $\simeq_f$ defined as $\simeq_f = \{(s,t) \mid f(s) = f(t)\}$, which is label-preserving iff $f$ is label-preserving. 
\end{definition}

We reduce the problem of identifying a suitable state classifier and ranking function to finding appropriate parameters for parametric function templates $f \colon \Theta \times S \to C$ and $r \colon H \times S \times S \to \mathbb{N}$. These templates define mappings that are fully determined by the parameters $\theta \in \Theta$ and $\eta \in H$, where $\Theta$ and $H$ are arbitrary parameter spaces. A state classifier template is label-preserving if the induced classifier is label-preserving for any parameterisation $\theta \in \Theta$. The specific parametric function templates used in our procedure are discussed in Section~\ref{sec:templates}. We write $f_\theta(s)$ for $f(\theta, s)$ and $r_\eta(s, s')$ for $r(\eta, s, s')$.

Since state classifiers over a finite set of classes correspond to finite partitions, Theorem~\ref{thm:boundednamjoshi} extends directly to this setting. Additionally, parametric function templates enable us to express the problem in first-order logic by shifting the focus from reasoning about the existence of suitable functions to reasoning about the existence of suitable parameters. We formalise this in the following corollary.

\begin{corollary}
    \label{col:boundednamjoshi}
    Let ${\cal M}$ be a labelled transition system with $k$-bounded branching, 
    $f \colon \Theta \times S \to C$ be a label-preserving state classifier template and $r \colon H \times S \times S \to \mathbb{N}$ be a ranking function template. 
    Let 
    \begin{align}
        \Psi(\theta,\eta, s,t) = \bigwedge_{i = 1}^k \Big( &\bigvee_{j = 1}^k f_\theta(\sigma_i(s)) = f_\theta(\sigma_j(t)) \; \lor \label{eq:cfirst}\\
        &f_\theta(s) =  f_\theta(\sigma_i(s))  \land r_\eta\left(\sigma_i(s) ,\sigma_i(s) \right)< r_\eta(s,s) \; \lor \label{eq:csec}\\
        &\bigvee_{j = 1}^k f_\theta(t) = f_\theta(\sigma_j(t)) \land r_\eta\left(\sigma_i(s) ,\sigma_j(t)\right) < r_\eta(\sigma_i(s) ,t)\Big). \label{eq:cthird}
    \end{align}
    Suppose that 
    \begin{equation}
        \exists \theta \in \Theta,\eta \in H \; \forall s,t \in S \colon f_\theta(s) = f_\theta(t) \implies \Psi(\theta,\eta, s,t).
    \end{equation}
    Then, $\simeq_f$ is a stutter-insensitive bisimulation on $\cal M$. \qed
\end{corollary}

Corollary~\ref{col:boundednamjoshi} reformulates the conditions of Theorem~\ref{thm:boundednamjoshi} in terms of parametric function templates. Branching bisimulation learning seeks to identify suitable parameters $\theta$ and $\eta$ that induce a valid well-founded bisimulation.

\subsection{Procedure}
\label{sec:procedure}

Our procedure involves two interacting components, the \emph{learner} and the \emph{verifier} that implement a CEGIS loop~\cite{DBLP:conf/cav/AbateDKKP18, DBLP:conf/fmcad/AlurBJMRSSSTU13}. The learner proposes candidate parameters that define a classifier and a ranking function, satisfying the conditions of Corollary~\ref{col:boundednamjoshi} over a finite set of sample states. The verifier then checks whether these induced mappings satisfy the conditions across the entire state space. 
If the verifier confirms that the conditions hold globally, the procedure has successfully synthesised a valid stutter-insensitive bisimulation, as induced by the classifier. From this, the corresponding finite quotient system is extracted, which can be verified  using off-the-shelf finite-state model checkers, with results directly applicable to the original, potentially infinite system. 
If the induced mappings fail to generalise to the entire state space, the verifier generates a counterexample pair of states that violates the conditions. This counterexample is returned to the learner, which refines the parameters to eliminate the violation. An overview of the procedure is depicted in Figure~\ref{fig:cegis}, and we elaborate on its individual components in the following sections.

\begin{figure}[t]

\centering
\resizebox{1\textwidth}{!}{%
\begin{tikzpicture}

\begin{scope}[on background layer]

\end{scope}

\draw [ fill=lightred, line width=0.2pt , fill opacity = 0.5, text opacity = 1, rounded corners] 
    (8.35,11.03) rectangle  node (quotientsynth) {} (13.65,8.2);

\draw [ fill=BrickRed, line width=0.2pt , fill opacity = 0.2, text opacity = 1, rounded corners] 
    (8.5,10.25) rectangle  node (learner) {\scriptsize \textbf{Learner}} (10.5,9);

\draw [ fill=BrickRed, line width=0.2pt , fill opacity = 0.2, text opacity = 1, rounded corners] 
    (11.5,10.25) rectangle  node (certifier) {\scriptsize \textbf{Verifier}} (13.5,9);

\node[] at (8,10.8) {\scriptsize$f$};
\node[] at (7.8,8.4) {\scriptsize UNSAT};

\draw [ fill=lightred, line width=0.2pt , fill opacity = 0.8, text opacity = 1, rounded corners, align = center] 
    (5.5,10.25) rectangle  node (subquo) {\scriptsize \textbf{Classifier}\\[-3pt]\scriptsize \textbf{Template}} (7.5,9);

\draw[-stealth] ([yshift=3.7em, xshift=0em]subquo.north) to[bend left = 0] ([yshift=.9em, xshift = 0em]subquo.north);
\node[] at (6.5,11.4) {\scriptsize$\llangle \cdot \rrangle$};

\draw[-stealth] ([yshift=4.2em, xshift=0em]learner.north) to[bend left = 0] ([yshift=1.35em, xshift = 0em]learner.north);
\node[] at (9.5,11.4) {\scriptsize$D, \mathcal{M}, r$};

\node[] at (12.5,11.4) {\scriptsize$\mathcal{M}, r$};
\draw[-stealth] ([yshift=4.2em, xshift=0em]certifier.north) to[bend left = 0] ([yshift=1.35em, xshift = 0em]certifier.north);

\draw[-stealth] ([yshift=.85em, xshift=1em]subquo.north) to[bend left = 27] ([yshift=1.4em, xshift = -1em]learner.north);
\draw[-stealth] ([yshift=-1.3em, xshift=-1em]learner.south) to[bend left = 27] ([yshift=-0.9em, xshift = 1em]subquo.south);

\draw[-stealth] ([yshift=1.3em, xshift=1em]learner.north) to[bend left = 27] ([yshift=1.4em, xshift = -1em]certifier.north);
\draw[-stealth] ([yshift=-1.3em, xshift=-1em]certifier.south) to[bend left = 27] ([yshift=-1.4em, xshift = 1em]learner.south);

\node[] at (11,10.8) {\scriptsize$\hat{\theta}, \hat{\eta}, f$};
\node[] at (11,8.45) {\scriptsize$(\hat{s}_{cex},\hat{t}_{cex})$};
\node[] at (11.1,7.8) {\footnotesize Branching Bisimulation Learning};

\draw [ fill=lightred, line width=0.2pt , fill opacity = 0.8, text opacity = 1, rounded corners, align = center] 
    (14.75,10.25) rectangle node (quotient) {\scriptsize \textbf{Quotient}\\[-3pt]\scriptsize \textbf{Extraction}} (16.75,9);

\draw [ fill=Apricot, line width=0.2pt , fill opacity = 0.4, text opacity = 1, rounded corners, align = center] 
    (17.75,10.25) rectangle  node (modelchecking) {\scriptsize \textbf{Model}\\[-3pt] \scriptsize \textbf{Checking}} (19.75,9);

\draw[-stealth] ([yshift=0em, xshift=1.15em]certifier.east) to[bend left = 0] ([yshift=0em, xshift = -0.5em]quotient.west);
\draw[-stealth] ([yshift=0em, xshift=0.48em]quotient.east) to[bend left = 0] ([yshift=0em, xshift = -0.75em]modelchecking.west);

\draw[-stealth] ([yshift=3.6em, xshift=0em]modelchecking.north) to[bend left = 0] ([yshift=.9em, xshift = 0em]modelchecking.north);

\node[align = center] at (18.75,11.4) {\footnotesize $\phi$};
\node[align = center] at (14.2,9.9) {\scriptsize $f(\hat{\theta}; \cdot)$};
\node[align = center] at (17.25,9.85) {\scriptsize $\quot$};

\draw[-stealth] ([yshift=-.8em, xshift=-.85em]modelchecking.south) to[bend left = 0] ([yshift=-2.4em, xshift = -1.6em]modelchecking.south);
\draw[-stealth] ([yshift=-.8em, xshift=.85em]modelchecking.south) to[bend left = 0] ([yshift=-2.4em, xshift = 1.6em]modelchecking.south);

\node[align = center] at (17.95,8.1) {\footnotesize \cmark \\[-4pt] \scriptsize Property\\[-3pt] \scriptsize Satisfied};
\node[align = center] at (19.5,8.08) {\footnotesize \xmark \\[-4pt] \scriptsize Counter-\\[-3.5pt] \scriptsize example};

\end{tikzpicture}
}

\caption{Architecture of branching bisimulation learning.}
\label{fig:cegis}
\end{figure}
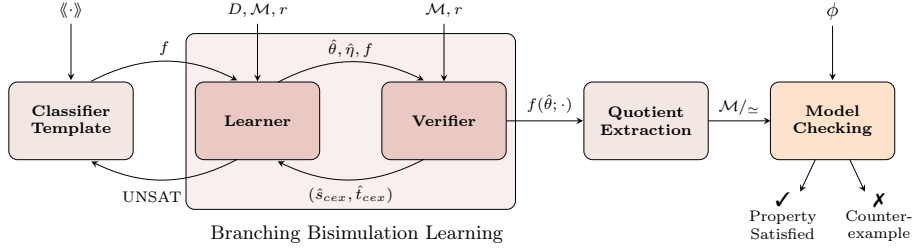

\subsubsection{Learner} The learner aims to find suitable parameters for a label-preserving state classifier template $f \colon \Theta \times S \to C$ and a ranking function template $r \colon H \times S \times S \to \mathbb{N}$ that satisfy the conditions of Corollary~\ref{col:boundednamjoshi} over a finite set of state pairs $D \subseteq S \times S$. Specifically, it attempts to solve: 

\begin{equation}
    \exists \theta \in \Theta,\eta \in H \colon \bigwedge_{(\hat{s}, \hat{t}) \in D} f_{{\theta}}(\hat{s}) = f_{{\theta}}(\hat{t}) \implies \Psi(\theta, \eta, \hat{s}, \hat{t}).  \label{eqn:learner}
\end{equation}

In our instantiation of the procedure, the learner is an SMT solver that seeks a satisfying assignment for the parameters $\theta$ and $\eta$ within the quantifier-free inner formula of~\eqref{eqn:learner}.
If the learner successfully identifies parameters that satisfy the conditions over the sample states, the resulting classifier and ranking function are passed to the verifier. However, if the learner fails to find suitable parameters, this indicates that the current function templates cannot be instantiated to comply with Corollary~\ref{col:boundednamjoshi} for the finite set of state pairs $D$.
This failure may arise for two reasons: First, since the model checking problem for infinite state systems is generally undecidable, the concrete system might not admit a finite stutter-insensitive bisimulation quotient, meaning no classifier or ranking function can satisfy the conditions of Corollary~\ref{col:boundednamjoshi}. Second, if a finite quotient does exist, the employed templates lack the expressiveness required to represent it. In this case, we must choose more expressive templates and continue the synthesis loop. In Section~\ref{sec:templates}, we detail how our instantiation of the procedure automatically increases the expressiveness of the templates as needed.

Equation~\eqref{eqn:learner} employs a single ranking function parameter $\eta \in H$ for the entire state space. To enhance flexibility, the ranking function can instead be defined piecewise with multiple parameters $\boldsymbol{\eta} = (\eta_c)_{c \in C}$, effectively assigning a separate ranking function to each class $c \in C$. While logically equivalent, this approach may enable the use of simpler templates by exploiting similarities in the temporal behavior of equivalent states. The learner then seeks to solve:
\begin{equation}
    \exists \theta \in \Theta,\boldsymbol{\eta} \in H^{|C|} \colon \bigwedge_{c \in C} \bigwedge_{(\hat{s}, \hat{t}) \in D} f_{{\theta}}(\hat{s}) = f_{{\theta}}(\hat{t}) = c \implies \Psi(\theta, \eta_c, \hat{s}, \hat{t}). \label{eqn:learner2}
\end{equation}
\subsubsection{Verifier}
The verifier checks whether the functions induced by the candidate parameters $\hat{\theta}$ and $\hat{\eta}$, proposed by the learner, generalise to the entire state space. For this, it attempts to solve the negation of the learner formula~\eqref{eqn:learner} for a counterexample pair of states:
\begin{equation}
    \exists s,t \in S \colon f_{\hat{\theta}}(s) = f_{\hat{\theta}}(t) \land \neg \Psi(\hat{\theta}, \hat{\eta}, {s}, {t}). 
    \label{eqn:cert}
\end{equation}
Similar to the learner, the verifier utilises an SMT solver, to which we provide the quantifier-free inner formula of~\eqref{eqn:cert}. If a satisfying assignment for a counterexample pair $(\hat{s}, \hat{t})$ is found, it is added to $D$ and returned to the learner, which refines the induced mappings to eliminate the counterexample. If the formula is unsatisfiable, this proves that Corollary~\ref{col:boundednamjoshi} holds for the entire state space, and the synthesis loop terminates with a valid stutter-insensitive bisimulation $\simeq_f$. 

If the learner employs a piecewise-defined ranking function as per Equation~\eqref{eqn:learner2}, the verifier instead checks the satisfiability of:
\begin{equation}
    \exists s, t \in S \colon \bigvee_{c \in C} f_{\hat{\theta}}(s) = f_{\hat{\theta}}(t) = c \land \neg \Psi(\hat{\theta}, \hat{\eta}_c, s, t),
    \label{eqn:cert2}
\end{equation}
for the candidate parameters $\hat{\theta}$ and $\boldsymbol{\hat{\eta}} = (\hat{\eta}_c)_{c \in C}$. The disjunction over the finite set of classes $C$ can be treated as independent and parallelisable SMT queries.

\subsubsection{Quotient Extraction} The learner-verifier framework yields a state classifier $f_{\hat{\theta}}(\cdot)$ that induces a valid stutter-insensitive bisimulation $\simeq_f$, provided it terminates successfully. From this, the corresponding finite quotient $\quot$ is derived by constructing the abstract transition relation $\qtrans$ and the initial states $\qinit$. The abstract state space $\qstsp$ corresponds to the finite set of classes $C$. 

To construct the abstract transition relation, we express the conditions from Definition~\ref{def:quotient} as quantifier-free first-order logic formulas, enabling efficient evaluation by an SMT solver. We then perform a series of independent, parallelisable queries for each pair of abstract states $c, d \in C$. 

\begin{itemize}
    \item An abstract transition \(c \qtrans d\) where \(c \neq d\) is established if:
    \begin{align}
        \exists s \in S \colon f_{\hat{\theta}}(s) = c \land \bigvee_{i = 1}^k f_{\hat{\theta}}( \sigma_i(s)) = d. \label{eq:transition}
    \end{align}
    \item An abstract transition \(c \qtrans c\) is established if:
    \begin{align}
        \nexists s \in S \colon f_{\hat{\theta}}( s) = c \land \bigwedge_{i = 1}^k f_{\hat{\theta}}(\sigma_i(s)) \neq c. \label{eq:selfloop}
    \end{align}
\end{itemize}
Both types of queries in Equations~\eqref{eq:transition} and~\eqref{eq:selfloop} are evaluated by passing the quantifier-free inner formula with the free variable \(s \in S\) to an SMT solver. Note that a self-loop \(c \qtrans c\) is established if the solver \emph{cannot} find a state \(s\) classified to \(c\) where all successors leave the class. This implies that every state in \(c\) has at least one successor within the same class.

To extract the abstract initial region $\qinit$, a single SMT query is sufficient for each abstract state. An abstract state is initial $c \in \qinit$ if and only if:
\begin{equation}
    \exists s \in S \colon f_{\hat{\theta}}(s) = c \wedge s \in I.
\end{equation}
Note that the synthesised ranking functions are not required for extracting the quotient. They are auxiliary in the synthesis of a valid stutter-insensitive bisimulation, while the resulting quotient depends solely on the final partition.

\section{Bisimulation Learning with Binary Decision Trees}
\label{sec:templates}
In this section, we detail our instantiation of the bisimulation learning procedure for non-deterministic systems.
We define parametric function templates for state classifiers and ranking functions employed in branching bisimulation learning, focusing on systems with discrete integer state spaces \(S \subseteq \mathbb{Z}^n\). 

For state classifiers, we use binary decision tree templates with parametric decision nodes and leaves corresponding to a finite set of classes~\cite{DBLP:conf/cav/AbateGS24}.

\begin{definition}[Binary Decision Tree Templates]
The set of binary decision tree (BDT) templates $\mathbb{T}$ over a finite set of classes $C$ and parameters $\Theta$ consists of trees $T$, which are defined as either:
\begin{itemize}
    \item a leaf node $\textsc{leaf}(c)$, where $c \in C$, or
    \item a decision node $\textsc{node}(\mu, T_l, T_r)$, where $T_l, T_r \in \mathbb{T}$ are the left and right subtrees, and $\mu \colon \Theta \times S \to \mathbb{B}$ is a parameterised predicate on the states.
\end{itemize}
A parametric tree template $T \in \mathbb{T}$ over classes $C$ and parameters $\Theta$ defines the state classifier template $f^T \colon \Theta \times S \to C$ as:
\begin{align*}
    f^T_{\theta}(s) = \begin{cases}
        c & \text{if } T = \textsc{leaf}(c), \\
        f^{T_l}_{\theta}(s) & \text{if } T = \textsc{node}(\mu, T_l, T_r) \text{ and } \mu_{\theta}(s), \\
        f^{T_r}_{\theta}(s) & \text{if } T = \textsc{node}(\mu, T_l, T_r) \text{ and } \neg \mu_{\theta}(s).
    \end{cases}
\end{align*}
\end{definition}
Binary decision trees are well-suited as state classifier templates due to their expressivity, interpretability, and straightforward translation into quantifier-free expressions over states and parameters.

To ensure that BDT templates are label-preserving, i.e., the induced classifiers respect the labelling of the original system for any parameterisation $\theta \in \Theta$, we associate atomic propositions $p \in \Pi$ with predicates $\mu_p \colon S \to \mathbb{B}$, such that:
\begin{equation}
    \obs{s} = \{p \in \Pi \mid \mu_p(s)\}.
\end{equation}
Label preservation is enforced by fixing parameter-free predicates for observations at the top nodes of the tree, with parametric decision nodes placed below them to refine the observation partition. This follows a similar principle as partition refinement~\cite{DBLP:journals/siamcomp/PaigeT87}, which starts from the observation partition and iteratively refines it.
In our instantiation, we use BDTs with affine predicates of the form \(\mu_{\theta}(s) := \theta_1 \cdot s + \theta_2 \leq 0\) at each parametric decision node, where \(\theta\) is drawn from the reals with appropriate dimension.  

BDTs facilitate an automatic increase of expressivity when the learner cannot find parameters that satisfy the well-founded bisimulation conditions over the finite set of sample states (cf. Section~\ref{sec:procedure}). The initial BDT templates are built automatically from the provided observation partition and include parameter-free top nodes that fix the labelling and a single layer of parametric decision nodes below. If suitable parameters cannot be found, we add another layer of parametric decision nodes, allowing further refinement of each class.

For parametric ranking function templates, we use affine functions of the form \(r_{\eta}(s, t) = \eta_1 \cdot s + \eta_2 \cdot t + \eta_3\), where \(\eta\) is an integer vector of appropriate dimension. These parameters must define a function that is bounded from below on its domain. For this, some systems may require a piecewise-defined ranking function.
Using affine predicates ensures that all conditions remain within a decidable fragment of first-order logic with linear arithmetic and efficient solution methods. While more expressive templates, such as non-linear ones, can be used, they require solving more complex SMT problems involving non-linear arithmetic, which are computationally expensive and, in some cases, undecidable~\cite{DBLP:series/txtcs/KroeningS16}.

\subsection{An Illustrative Example}
We illustrate our instantiation of branching bisimulation learning for the non-deterministic example program in Figure~\ref{fig:illustrative example}a.
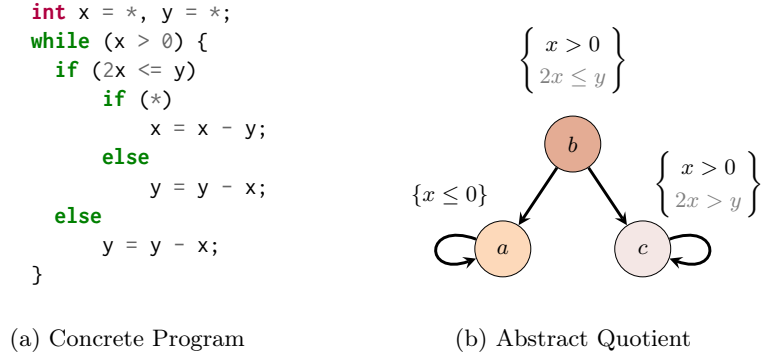
\begin{figure}[t]
    \centering  
    \begin{tabular}{ccc}
        \begin{minipage}[b]{0.47\columnwidth}
        \centering
        \begin{minted}[mathescape]{java}
          int x = *, y = *;
          while (x > 0) {
            if (2x <= y) 
                if (*)
                    x = x - y;
                else 
                    y = y - x;
            else 
                y = y - x;
          }
        \end{minted}

        \vfill
        \end{minipage}
        &
        \hspace{-100pt}
        &
        \begin{minipage}[b]{0.47\columnwidth}
        \centering
        \raisebox{12pt}{
        \scalebox{0.93}{
 \hspace{8pt}\begin{tikzpicture}[>=stealth]

\begin{scope}[scale=1.0, yshift=1cm]

\node[circle, draw, fill=BrickRed!35, align=center, 
      minimum size=8mm, inner sep=0pt] 
      (b) at (0,1.5) {$b$};

\node[circle, draw, fill=Apricot!50, align=center, 
      minimum size=8mm, inner sep=0pt] 
      (a) at (-1,0) {$a$};

\node[circle, draw, fill=lightred!80, align=center, 
      minimum size=8mm, inner sep=0pt] 
      (c) at (1,0) {$c$};

\node[align=center, above left=0.2cm and -0.2cm of a] 
      (label-a) {$\{x \leq 0\}$};

\node[align=center, above=0.2cm of b, font=\small] 
      (label-b) {$\left\{ \begin{array}{c} x > 0 \\ \textcolor{gray}{2x \leq y} \end{array} \right\}$};

\node[align=center, above right=0.05cm and -0.25cm of c, font=\small] 
      (label-c) {$\left\{ \begin{array}{c} x > 0 \\ \textcolor{gray}{2x > y} \end{array} \right\}$};

\draw[->, very thick] (b) -- (a);

\draw[->, very thick] (b) -- (c);

\draw[->, very thick] (a) edge[out=160, in=200, loop, min distance=8mm] node[midway, left] {} (a);

\draw[->, very thick] (c) edge[out=20, in=-20, loop, min distance=8mm] node[midway, right] {} (c);

\end{scope}
\end{tikzpicture}
}}
    \end{minipage}\\
    (a) Concrete Program && 
    (b) Abstract Quotient
    \end{tabular}
    \caption{A non-deterministic program and the corresponding stutter-insensitive bisimulation quotient synthesised with branching bisimulation learning.}
    \label{fig:illustrative example}
\end{figure}
The program takes two arbitrary integers \(x\) and \(y\) as inputs and iterates while \(x > 0\). Based on variable values and a non-deterministic choice, it subtracts either \(x\) from \(y\) or vice versa, inducing an infinite non-deterministic transition system over \(S = \mathbb{Z}^2\). States are labelled by \(x \leq 0\) and \(x > 0\), indicating whether computation has terminated or is still running. Despite having infinitely many states, the program has bounded branching, as the single non-deterministic choice yields at most two successors.

Determining which states have terminating or diverging branches is non-trivial. We use branching bisimulation learning to identify conditions under which states share identical computation trees up to exact step counts and extract the stutter-insensitive bisimulation quotient. Our initial BDT template fixes the labelling at a single non-parametric top node with predicate \(x \leq 0\) and includes a parametric decision node to refine the partition labelled \(x > 0\).

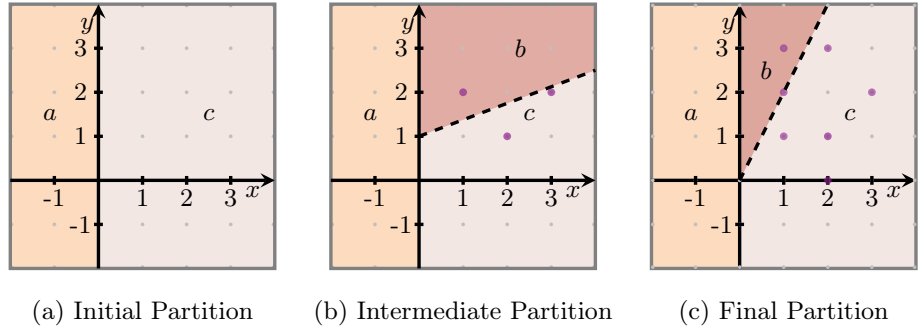
\begin{figure}[t]
\centering
\resizebox{\textwidth}{!}{%
\begin{tikzpicture}[>=stealth]

\begin{scope}[xshift=0cm, scale=0.55]

    \fill[Apricot, fill opacity=0.5]
        (-2,-2) -- (0,-2) -- (0,4) -- (-2,4) -- cycle;

    \fill[lightred, fill opacity=0.8]
        (0,-2) -- (4,-2) -- (4,4) -- (0,4) -- cycle;

    \node at (-1.1, 1.5) {$a$}; 
    \node at (2.5, 1.5) {$c$}; 

    \foreach \x in {-1.97,-1,0,1,2,3,3.97} {
        \foreach \y in {-1.97,-1,0,1,2,3,3.97} {
            \fill[lightgray] (\x,\y) circle (1.2pt);
        }
    }

    \draw[gray, very thick] (-2,-2) rectangle (4,4);

    \draw[very thick,-stealth] (-2,0) -- (4,0);
    \draw[very thick,-stealth] (0,-2) -- (0,4);

    \node at (3.5, -0.25) {$x$};
    \node at (-0.25, 3.5) {$y$};

    \foreach \x in {-1,1,2,3} {
        \draw[very thick] (\x, 0.1) -- (\x, -0.1);
        \node[font=\small] at (\x, -0.4) {\x};
    }
    \foreach \y in {-1,1,2,3} {
        \draw[very thick] (0.1,\y) -- (-0.1,\y);
        \node[font=\small] at (-0.4,\y) {\y};
    }

    \node at (1, -3) {(a) Initial Partition};
\end{scope}

\begin{scope}[xshift=4cm, scale=0.55]

    \fill[Apricot, fill opacity=0.5]
        (-2,-2) -- (0,-2) -- (0,4) -- (-2,4) -- cycle;

    \fill[lightred, fill opacity=0.8]
        (0,-2) -- (4,-2) -- (4,2.5) -- (0,1) -- cycle;

    \fill[BrickRed, fill opacity=0.4]
        (0,1) -- (4,2.5) -- (4,4) -- (0,4) -- cycle;

    \node at (-1.1, 1.5) {$a$}; 
   \node at (2.3, 3) {$b$}; 
 \node at (2.5, 1.5) {$c$}; 

    \foreach \x in {-1.97,-1,0,1,2,3,3.97} {
        \foreach \y in {-1.97,-1,0,1,2,3,3.97} {
            \fill[lightgray] (\x,\y) circle (1.2pt);
        }
    }

    \draw[gray, very thick] (-2,-2) rectangle (4,4);

    \draw[very thick,-stealth] (-2,0) -- (4,0);
    \draw[very thick,-stealth] (0,-2) -- (0,4);

    \node at (3.5, -0.25) {$x$};
    \node at (-0.25, 3.5) {$y$};

    \foreach \x in {-1,1,2,3} {
        \draw[very thick] (\x, 0.1) -- (\x, -0.1);
        \node[font=\small] at (\x, -0.4) {\x};
    }
    \foreach \y in {-1,1,2,3} {
        \draw[very thick] (0.1,\y) -- (-0.1,\y);
        \node[font=\small] at (-0.4,\y) {\y};
    }

    \fill[Plum, opacity=0.7] (2,1) circle (2.5pt);
    \fill[Plum, opacity=0.7] (3,2) circle (2.5pt);
    \fill[Plum, opacity=0.7] (1,2) circle (2.5pt);

    \draw[very thick, dashed] (0,1) -- (4,2.5);

    \node at (1, -3) {(b) Intermediate Partition};
\end{scope}

\begin{scope}[xshift=8cm, scale=0.55]

    \fill[Apricot, fill opacity=0.5]
        (-2,-2) -- (0,-2) -- (0,4) -- (-2,4) -- cycle;

    \fill[lightred, fill opacity=0.8]
        (0,-2) -- (4,-2) -- (4,4) -- (0,4) -- cycle;

    \fill[BrickRed, fill opacity=0.35]
        (0,0) -- (2,4) -- (0,4) -- cycle;

    \node at (-1.1, 1.5) {$a$}; 
    \node at (0.6, 2.5) {$b$}; 
 \node at (2.5, 1.5) {$c$}; 

    \draw[gray, very thick] (-2,-2) rectangle (4,4);

    \draw[very thick,-stealth] (-2,0) -- (4,0);
    \draw[very thick,-stealth] (0,-2) -- (0,4);

    \node at (3.5, -0.25) {$x$};
    \node at (-0.25, 3.5) {$y$};

    \foreach \x in {-1.97,-1,0,1,2,3,3.97} {
        \foreach \y in {-1.97,-1,0,1,2,3,3.97} {
            \fill[lightgray] (\x,\y) circle (1.2pt);
        }
    }
    
    \foreach \x in {-1,1,2,3} {
        \draw[very thick] (\x, 0.1) -- (\x, -0.1);
        \node[font=\small] at (\x, -0.4) {\x};
    }
    \foreach \y in {-1,1,2,3} {
        \draw[very thick] (0.1,\y) -- (-0.1,\y);
        \node[font=\small] at (-0.4,\y) {\y};
    }

    \fill[Plum, opacity=0.7] (2,1) circle (2.5pt);
    \fill[Plum, opacity=0.6] (3,2) circle (2.5pt);
    \fill[Plum, opacity=0.6] (1,2) circle (2.5pt);
    \fill[Plum, opacity=0.6] (2,0) circle (2.5pt);
    \fill[Plum, opacity=0.6] (1,1) circle (2.5pt);
    \fill[Plum, opacity=0.6] (2,3) circle (2.5pt);
    \fill[Plum, opacity=0.6] (1,3) circle (2.5pt);

    \draw[very thick, dashed] (0,0) -- (2,4);

    \node at (1, -3) {(c) Final Partition};
\end{scope}

\end{tikzpicture}
}
\caption{Iterative process of branching bisimulation learning, illustrating three stages with generated counterexample states (purple dots).}
\label{fig:blprocess}
\end{figure}

The process of branching bisimulation learning is illustrated in Figure~\ref{fig:blprocess}. Starting from an arbitrary initial partition (Fig.~\ref{fig:blprocess}a), the procedure iteratively generates counterexample states and refines the partition to meet the conditions of a well-founded bisimulation over the finite sample set (Fig.~\ref{fig:blprocess}b). The process terminates when no counterexample states remain (Fig.~\ref{fig:blprocess}c). 

The resulting partition certifies that refining the non-terminated class labeled \(x > 0\) along the predicate \(2x - y \leq 0\) induces a valid stutter-insensitive bisimulation. From this partition, we extract the stutter-insensitive bisimulation quotient, shown in Figure~\ref{fig:illustrative example}b. Initial states are omitted, as branching bisimulation learning is independent of initial states. It is stronger in that it determines the conditions defining a state's characteristics, as reflected by its outgoing computation tree, for all states simultaneously. In the quotient, the class \(b\) precisely represents states with both terminating and diverging outgoing paths, as expressed by the \(\ctln\) formula \(\phi = \exists \lozenge \square (x \leq 0) \land \exists \square (x > 0)\).

In this example, the initial BDT template was sufficiently expressive to capture a valid stutter-insensitive bisimulation on the state space. If the learner fails to find a partition that satisfies the conditions in Section~\ref{sec:procedure} over the finite set of sample states, BDTs enable an automatic increase in expressivity by adding an extra layer of decision nodes which further refine the partition.

\section{Experimental Evaluation}
\label{sec:experiments}
We implemented our instantiation of branching bisimulation learning with BDT classifier templates in a software prototype and evaluated it across a diverse range of case studies, including deterministic and concurrent software, communication protocols, and reactive systems. Our procedure is benchmarked against state-of-the-art tools, including the nuXmv model checker~\cite{DBLP:conf/vmcai/Bradley11,DBLP:journals/iandc/BurchCMDH92,DBLP:conf/cav/CavadaCDGMMMRT14}, the Ultimate~\cite{DBLP:conf/tacas/HeizmannBDFHKNSSP23} and CPAChecker~\cite{DBLP:conf/cav/BeyerK11} software verifiers, and the T2~\cite{DBLP:conf/tacas/BrockschmidtCIK16} infinite-state branching-time model checker. 
We employ the Z3 SMT solver~\cite{DBLP:conf/tacas/MouraB08} in both the learner and verifier components of the learning loop and use the nuXmv model checker to verify the obtained finite stutter-insensitive bisimulation quotients.

\subsubsection{Setup}  
Branching bisimulation learning provides a unified approach for verifying infinite-state non-deterministic systems with respect to \(\ctln\) specifications. Since our procedure encompasses a broad range of system and specification classes, we compare its performance against specialised state-of-the-art tools. Specifically, we evaluate its effectiveness on key special cases, including deterministic infinite-state and large finite-state systems, before extending our analysis to the full generality of infinite-state and non-deterministic systems. 
Concretely, we consider:
\begin{enumerate}
    \item Deterministic finite-state clock synchronisation protocols~\cite{DBLP:conf/cav/AbateGS24,DBLP:conf/ifip/Lamport83}, including TTEthernet~\cite{DBLP:conf/cpsweek/BogomolovHS16,DBLP:conf/tacas/BogomolovFGH17} and an interactive convergence algorithm~\cite{DBLP:conf/podc/LamportM84}, ensuring \sloppy agents synchronise despite clock drift. We verify safety (clocks stay within a safe distance) and liveness (agents repeatedly synchronise). We compare against nuXmv using IC3 and BDD-based symbolic model checking, and deterministic bisimulation learning~\cite{DBLP:conf/cav/AbateGS24}, which is tailored to this type of system. We evaluate multiple instances with varying time discretisations, where smaller time steps lead to larger state spaces.
    
    \item Conditional termination of infinite-state deterministic software from the SV-COMP termination category~\cite{DBLP:conf/cav/BeyerK11}. We compare against nuXmv using IC3, Ultimate, CPAChecker, and deterministic bisimulation learning. Here, non-bisimulation-based approaches require separate benchmark instances, as they can only verify termination or non-termination of all initial states. In contrast, bisimulation learning synthesises a quotient, precisely identifying conditions for termination and divergence (see Figure~\ref{fig:illustrative example}).

    \item  Infinite-state concurrent software and reactive systems with non-deterministic inputs. We consider non-deterministic conditional termination benchmarks, concurrent software from the T2 verifier benchmark set~\cite{DBLP:conf/cav/BeyenePR13,DBLP:conf/tacas/BrockschmidtCIK16}, and reactive robotics where agents navigate while avoiding collisions. For $\ltln$ specifications, we compare against nuXmv (IC3) and Ultimate. For $\ctlnn$ and $\ctln$, we compare against T2, the only tool for branching-time verification of infinite-state non-deterministic systems. We also consider finite-state versions verified for \(\ctlnn\) using nuXmv with BDDs.
    
\end{enumerate}

\begin{figure}[t]
    \centering
         \includegraphics[height=\scatterplotheight]{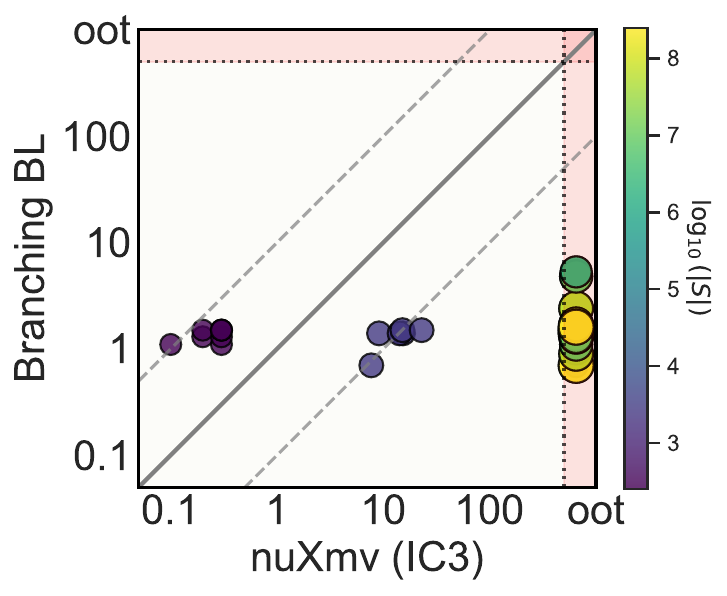}
        \includegraphics[height=\scatterplotheight]{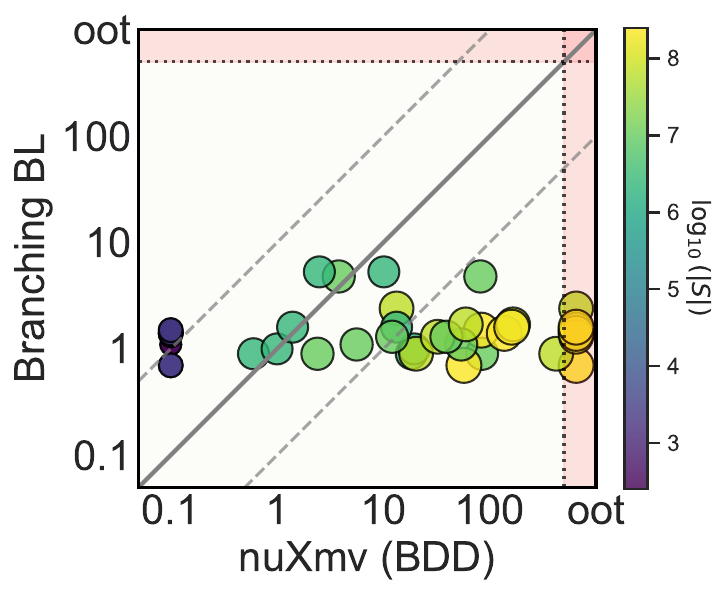}
        \includegraphics[height=\scatterplotheight]{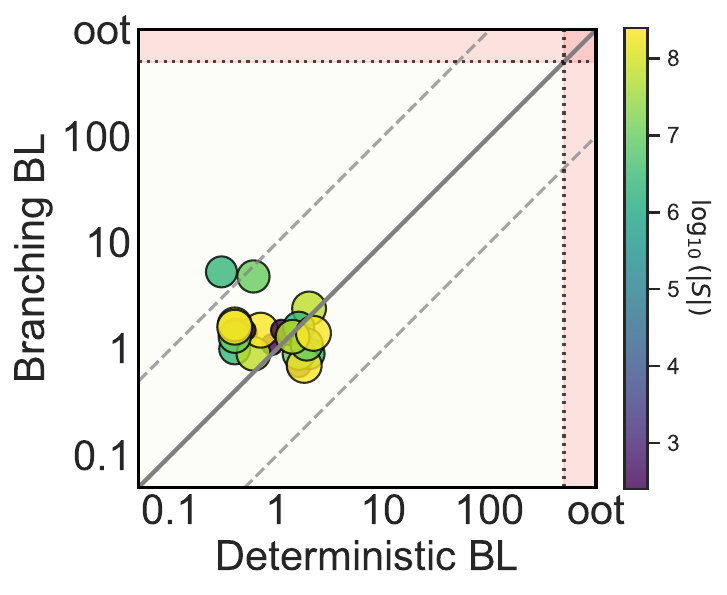}
    \caption{Deterministic finite-state clock synchronisation benchmarks, verified against $\ltln$ specifications. Data-point colours indicate state-space size. }
    \label{fig:clock}
\end{figure}
\begin{figure}[t]
    \centering
    \includegraphics[height=\scatterplotheight]{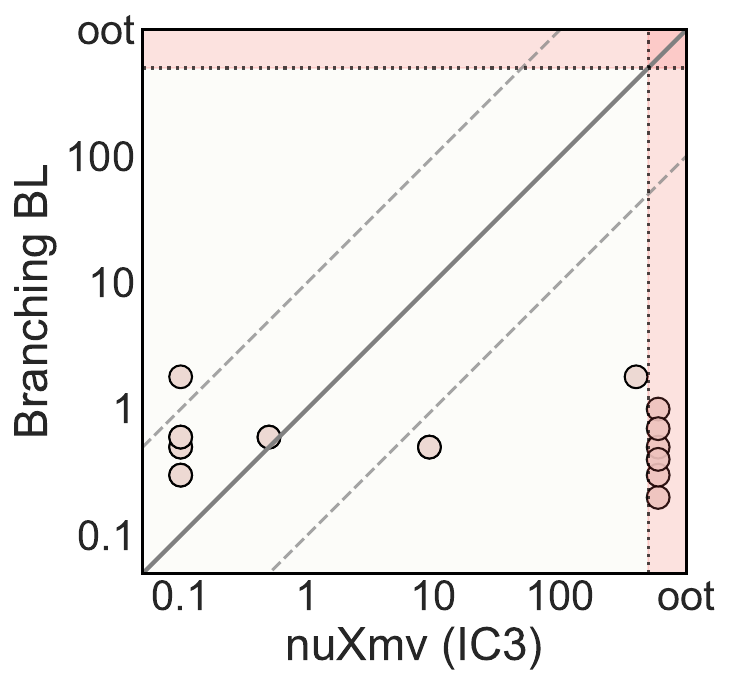}\hfill
    \includegraphics[height=\scatterplotheight]{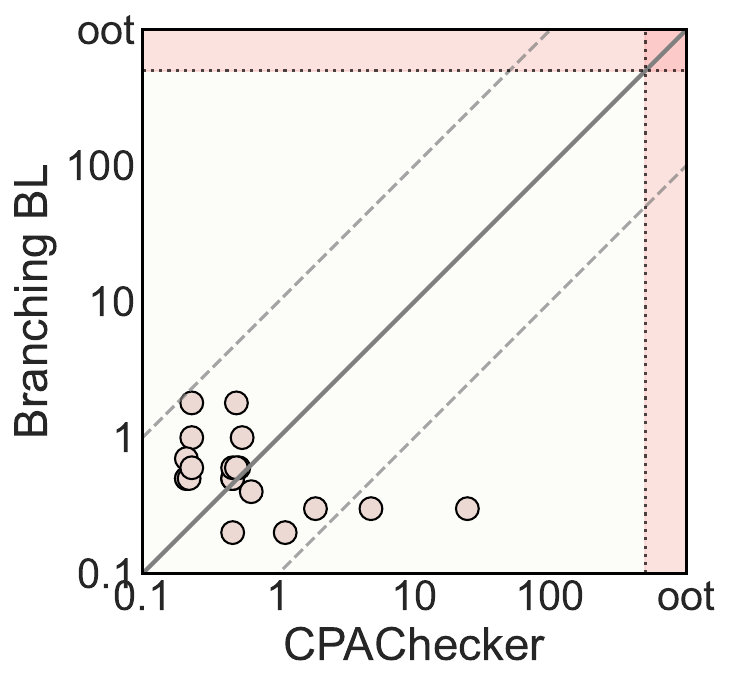}\hfill
    \includegraphics[height=\scatterplotheight]{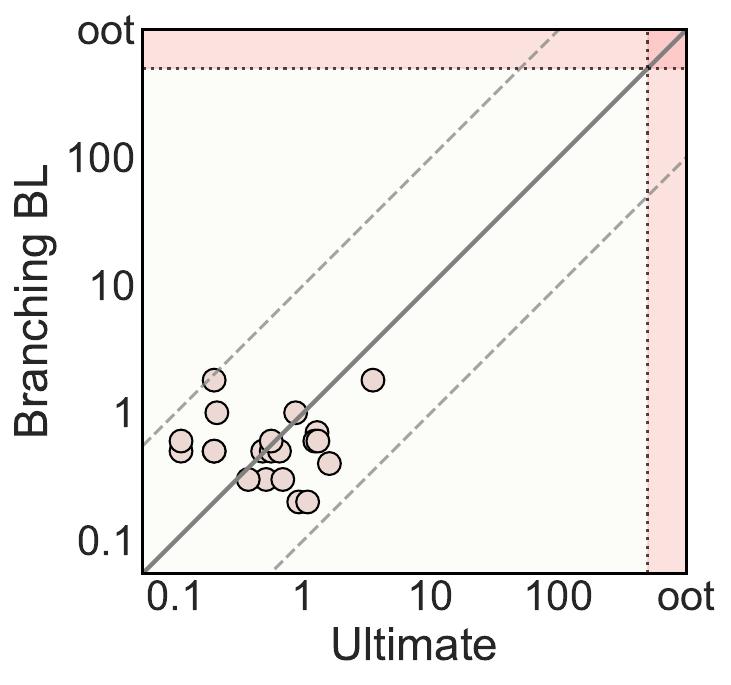}\hfill
    \includegraphics[height=\scatterplotheight]{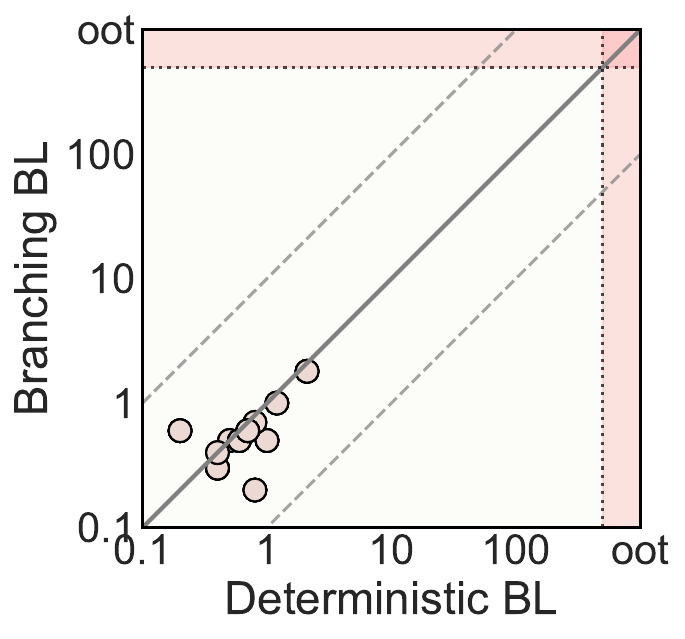}
    \caption{Deterministic infinite-state termination analysis benchmarks.}
    \label{fig:term}
\end{figure}
\begin{figure}[t]
    \centering
    \begin{tabular}{cccc}
    \;\;\;\;\; \scriptsize $\ltln$ & \;\;\;\;\;\;\scriptsize$\ltln$ & \;\;\;\;\;\scriptsize$\ctlnn$ & \;\;\;\;\; \scriptsize$\ctln$\\
        \includegraphics[height=\scatterplotheight]{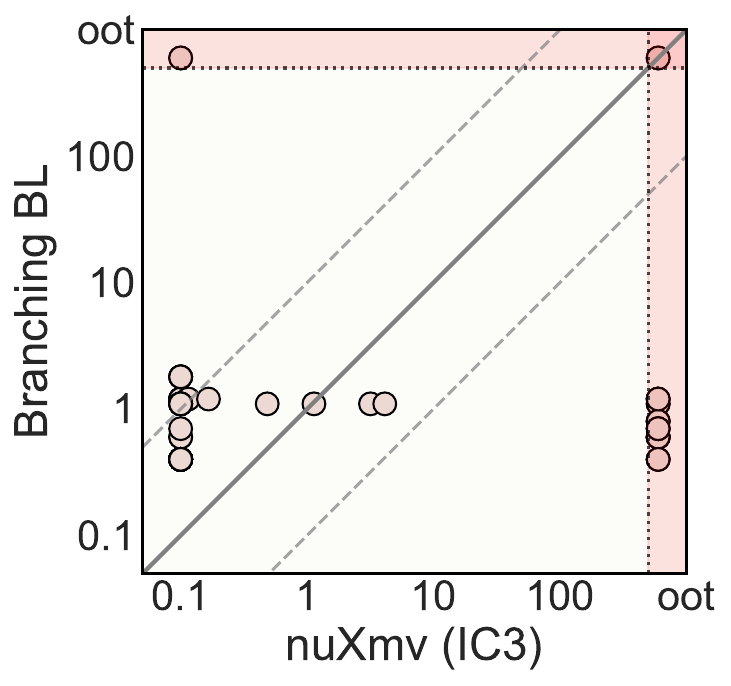}\hfill&
    \includegraphics[height=\scatterplotheight]{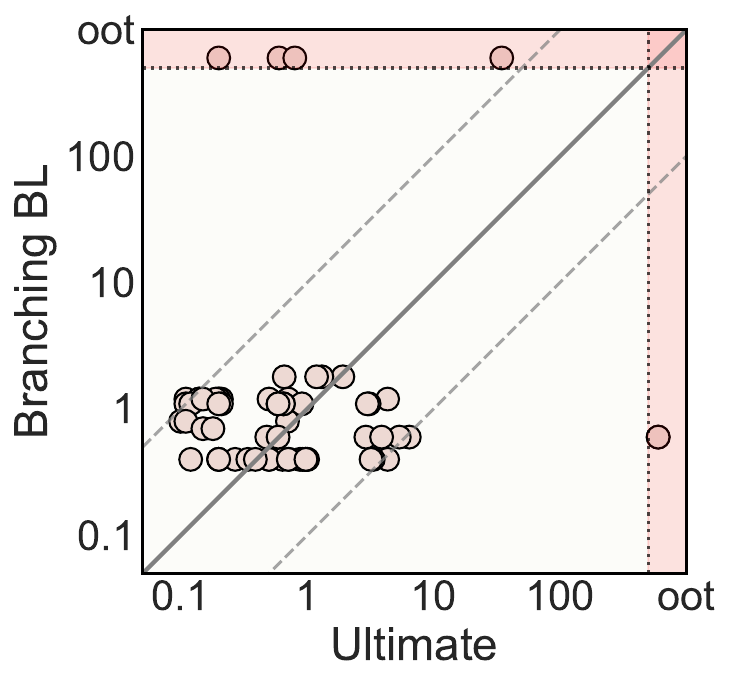}\hfill&
    \includegraphics[height=\scatterplotheight]{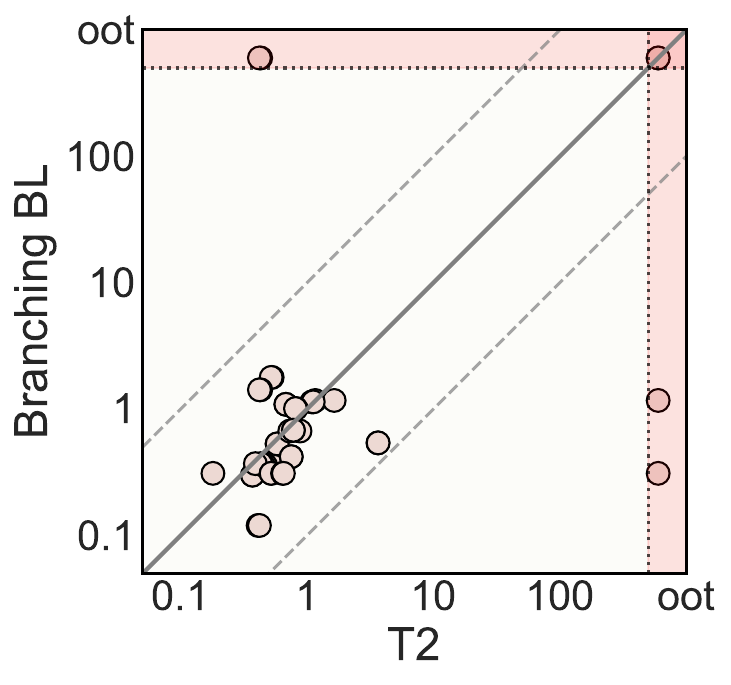}\hfill&
    \includegraphics[height=\scatterplotheight]{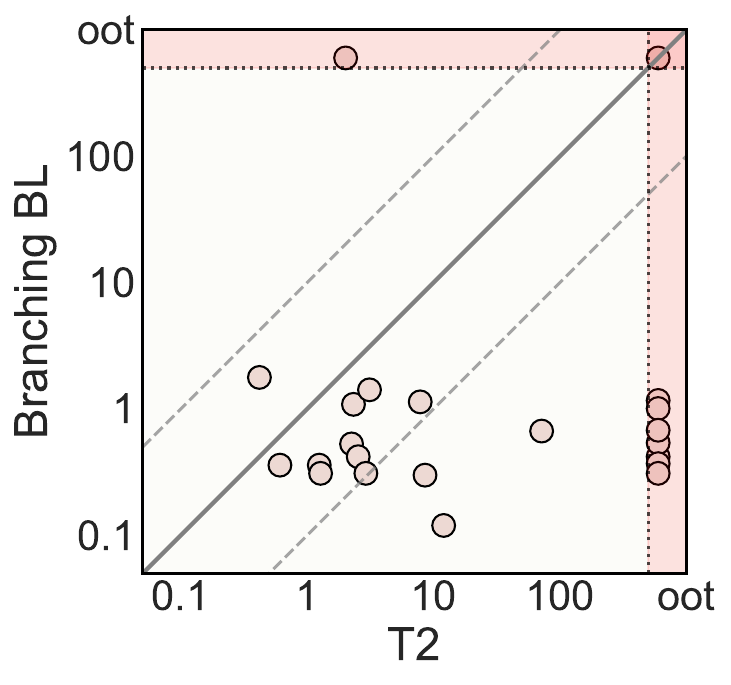}\hfill
    \end{tabular}
        \caption{Non-deterministic infinite-state benchmarks.}
    \label{fig:ltlctl}

\end{figure}
\subsubsection{Results}
We present the resulting verification runtimes in Figures~\ref{fig:clock} to~\ref{fig:ltlctl}\ifthenelse{\isundefined{\techreport}}{. Detailed results can be found in the technical report~\cite{arxiv}}{, with detailed results and verified formulas in Appendix~\ref{app:exp}}. Each figure compares branching bisimulation learning with BDT templates against applicable baseline tools. Runtimes are shown in seconds, with dashed diagonals indicating 10-fold differences. Each data point represents a case-study and formula combination. 

Timeouts in the red areas correspond to runtimes exceeding 500 seconds. For baselines, we report only analysis time, excluding any preprocessing of programs. Since bisimulation learning depends on the sequence of produced counterexamples, which is generally non-deterministic, we run each benchmark 10 times. The plots report average runtimes, with variances detailed in \ifthenelse{\isundefined{\techreport}}{the technical report~\cite{arxiv}}{Appendix~\ref{app:exp}}.

\subsubsection{Discussion}
The results demonstrate that branching bisimulation learning is an effective and general verification approach across key system and specifications classes, extending to the full generality of non-deterministic infinite-state systems with \(\ctln\) specifications, where T2 is the only competitor.

For deterministic clock synchronisation and termination benchmarks, branching bisimulation learning remains competitive with baseline tools and even deterministic bisimulation learning~\cite{DBLP:conf/cav/AbateGS24}, despite using a more expressive proof rule. This is presumably due to the separation of partition learning and quotient extraction, reducing the complexity of used SMT queries. As a result, it retains the advantages of deterministic bisimulation learning. For termination benchmarks, our approach achieves runtimes comparable to state-of-the-art tools while also producing interpretable binary decision trees that characterise the conditions under which states terminate, aiding system diagnostics and fault analysis.

For non-deterministic infinite-state benchmarks, the results confirm that branching bisimulation learning is effective in synthesising succinct \sloppy stutter-insensitive bisimulation quotients, enabling efficient \(\ctln\) verification by an off-the-shelf finite-state model checker. Notably, for pure \(\ctln\) properties beyond \(\ltln\) or \(\ctlnn\), our approach outperforms T2, the only available alternative, on its benchmark set.  This shows the advantage of branching bisimulation learning in making the verification of complex branching-time properties scalable and practical for non-deterministic infinite-state systems.

We note that bisimulation learning is not a stand-alone procedure but a lightweight preprocessing technique that synthesises succinct finite quotients for infinite-state systems, enabling their verification with standard finite-state model checkers. Since model checking for this general class of systems is undecidable, our approach is necessarily incomplete, as not all systems admit a finite stutter-insensitive bisimulation quotient. However, our experiments demonstrate that many standard benchmarks do yield finite quotients, and branching bisimulation learning is able to efficiently derive them from sampled system behaviours.

\section{Related Work}

Notions of abstractions, particularly bisimulation relations and their efficient computation, have been widely studied in the literature~\cite{DBLP:journals/fac/BalcazarGS92,DBLP:journals/iandc/KanellakisS90,DBLP:books/sp/Milner80,DBLP:books/cu/12/AcetoIS12}. 
The primary notion is \emph{strong} bisimulation, which requires related states to match each other’s transitions exactly at every step, thereby preserving all linear- and branching-time properties expressible in common specification logics such as $\ltl$, $\ctls$, and even the $\mu$-calculus~\cite{DBLP:conf/icalp/Kozen82}. This notion has been relaxed into weaker variants that preserve only the externally observable behaviour, abstracting from exact stepwise equivalence, such as \emph{stutter-insensitive} bisimulation~\cite{DBLP:journals/tcs/BrowneCG88,DBLP:journals/jacm/NicolaV95}. These variants allow for much more succinct quotients while still preserving properties that do not include the exact stepwise \emph{next}-operator, such as formulas in $\ctln$. 

Stutter-insensitive bisimulations are closely related to \emph{branching} bisimulations~\cite{DBLP:journals/jacm/GlabbeekW96}, which serve as the natural analogue when actions rather than states are labelled~\cite{DBLP:conf/icalp/GrooteV90}. Any stutter-insensitive bisimulation on a state-labelled transition system forms a branching bisimulation on the corresponding action-labelled transition system, with both representations interconvertible via a standard construction~\cite{DBLP:conf/icalp/GrooteV90,DBLP:journals/tocl/GrooteJKW17}. This observation motivates the name of our procedure, which extends the applicability of  bisimulation learning to nondeterministic systems and branching-time specifications, while also enabling the computation of branching bisimulations for action-labelled transition systems, even though our focus remains on state-labelled systems~\cite{DBLP:conf/cav/AbateGS24}.

The standard approach to computing bisimulation relations are partition refinement algorithms, such as the Paige-Tarjan algorithm~\cite{HOPCROFT1971189,DBLP:journals/siamcomp/PaigeT87,DBLP:conf/concur/FernandezM91}. These algorithms iteratively refine an equivalence relation with respect to the bisimulation conditions until the partition stabilises into a valid bisimulation. This generalises to stutter-insensitive and branching bisimulations, although it requires computing unbounded pre or post images of the transition relation to determine the refinement~\cite{DBLP:conf/icalp/GrooteV90,DBLP:journals/tocl/GrooteJKW17,DBLP:journals/corr/abs-1909-10824}. Since partition refinement must process the entire state space, it can fall short on large systems. Moreover, for infinite state systems, symbolic procedures are required, incurring costly quantifier elimination~\cite{DBLP:conf/rp/MummeC11,DBLP:conf/cav/BouajjaniFH90,DBLP:conf/concur/AlfaroHM01}. Bisimulation learning introduces an incremental approach based on inductive synthesis that computes stutter-insensitive bisimulations by generalising from sample transitions, thereby circumventing the need to process the entire state space and enabling efficient computation for infinite state systems~\cite{DBLP:conf/cav/AbateGS24,DBLP:conf/cav/AbateDKKP18,DBLP:conf/birthday/AbateGRS25}.

Model checking of infinite-state systems with respect to branching-time properties has been explored through techniques based on the satisfiability of Horn clauses~\cite{DBLP:conf/birthday/BjornerGMR15,DBLP:conf/cav/Gurfinkel22,DBLP:conf/cade/BjornerMR12,DBLP:conf/sas/BjornerMR13}, a class of universally quantified first-order logic formulas. These clauses are particularly effective for encoding safety properties of nondeterministic infinite-state systems. When extended with existential quantification, they facilitate the verification of full $\text{CTL}^*$ properties~\cite{DBLP:conf/cav/BeyenePR13,DBLP:journals/corr/BeyenePR16,DBLP:conf/atva/CarelliG24,DBLP:journals/tcs/KestenP05}. This approach is implemented in the T2 verification framework~\cite{DBLP:conf/tacas/BrockschmidtCIK16,DBLP:conf/cav/CookKP15}. However, existential Horn clauses introduce quantifier alternation, which is computationally expensive even for mature symbolic first-order logic solvers~\cite{DBLP:conf/pldi/SFMD21,DBLP:conf/cav/GarciaContrerasKSG23}.
An alternative approach reduces $\text{CTL}^*$ formulas to the modal $\mu$-calculus and verifies the resulting expression by solving a parity game~\cite{DBLP:journals/tcs/Dam94,DBLP:journals/tcs/NiwinskiW96}. However, the translation from $\ctls$ to the $\mu$-calculus incurs a doubly-exponential blowup in the size of the formula, rendering the approach infeasible for many practical examples.

\section{Conclusion}
\label{sec:concl}

We have presented a generalised branching bisimulation learning algorithm for non-deterministic systems with bounded branching, and 
enabled effective model checking of finite- and infinite-state systems against $\ctln$ specifications. 
Our lightweight approach implements the proof rule for well-founded bisimulations within a 
counterexample-guided inductive synthesis loop, which performs comparably to mature tools for the special case of linear-time properties and 
establishes a new state of the art for branching-time properties. 
Our method produces stutter-insensitive bisimulations of the system under analysis for proving branching- and linear-time temporal properties as well as providing genuine abstract counterexamples or the exact initial conditions for which the property is satisfied. 
Our contribution provides the groundwork for the development of efficient model checkers that integrate bisimulation learning as a technique to reduce non-deterministic systems with very large or infinite state spaces to equivalent systems with finite and succinct state space. 

\subsubsection*{\ackname}
This work was funded in part by the Advanced Research + Invention Agency (ARIA) under the Safeguarded AI programme. 

\subsubsection*{\discintname}
The authors have no competing interests to declare that are relevant to the content of this article.

%
%
\bibliographystyle{splncs04}
\bibliography{main}
\ifthenelse{\isundefined{\techreport}}{
}{
\appendix
\section{Detailed Experimental Evaluation}
\label{app:exp}

We provide details for the benchmarks, verified properties, and measured runtimes for the experiments presented in Section~\ref{sec:experiments}.

\vspace{-3pt}
\subsection{Deterministic Finite-state Clock Synchronisation}
Table~\ref{tab:clock} details the deterministic finite-state clock synchronisation benchmarks from Figure~\ref{fig:clock}. We consider two protocols:  
\begin{enumerate}
    \item TTEthernet: Agents send clock values to a central master, which computes and broadcasts the median~\cite{DBLP:conf/cpsweek/BogomolovHS16}.
    \item Interactive Convergence: Agents exchange clock values, averaging them while discarding outliers~\cite{DBLP:conf/ifip/Lamport83}. 
\end{enumerate}
We verify (i) a safety invariant ensuring clocks stay within a maximum distance and (ii) a liveness property ensuring periodic synchronisation. Each protocol is tested in safe (satisfying) and unsafe (violating) variants. Finer time dicretisation increases timesteps per second, expanding the state space.
\newcolumntype{N}{>{\centering\arraybackslash}m{2.2cm}}
\newcolumntype{M}{>{\centering\arraybackslash}m{1.5cm}}
\newcolumntype{D}{>{\centering\arraybackslash}m{2.4cm}}
\renewcommand{\arraystretch}{1.125}
\begin{table}[ht]
\vspace{-12pt}
\centering
\caption{
Results for reactive clock synchronisation benchmarks. Benchmark names encode parameters, e.g., "tte-sf-1k" denotes a safe TTEthernet instance with a 1000-step-per-second dicretisation. All times are measured in seconds.
}

\vspace{5pt}
\resizebox{\linewidth}{!}{
\begin{tabular}{ccMMMMDN}
\toprule
Benchmark & $|S|$ & \multicolumn{2}{c}{\textbf{nuXmv} (IC3)} & \multicolumn{2}{c}{\textbf{nuXmv} (BDDs)} & \multicolumn{2}{c}{\textbf{Bisimulation Learning}} \\
          &            & $\square(\text{safe})$ & $\square\lozenge(\text{sync})$ & $\square(\text{safe})$ & $\square\lozenge(\text{sync})$ & Deterministic & Branching \\
\midrule
tte-sf-10    & 250         & 0.1  & 0.3  & < 0.1 & < 0.1 & 0.9\tiny$\pm 0.5$ & 1.1\tiny$\pm 0.8$ \\
tte-sf-100   & 2500        & 7.7  & oot  & < 0.1 & 0.1  & 1.6\tiny$\pm 1.1$ & 0.7\tiny$\pm 0.3$ \\
tte-sf-1k    & \num{2.5e6} & oot  & oot  & 0.6  & 18.3 & 2.0\tiny$\pm 0.9$ & 0.9\tiny$\pm 0.2$ \\
tte-sf-2k    & \num{1e7}   & oot  & oot  & 2.4  & 84.6 & 1.6\tiny$\pm 1.0$ & 0.9\tiny$\pm 0.2$ \\
tte-sf-5k    & \num{6.25e7}& oot  & oot  & 13.3 & oot  & 2.0\tiny$\pm 0.9$ & 2.4\tiny$\pm 5.0$ \\
tte-sf-10k   & \num{2.5e8} & oot  & oot  & 57.2 & oot  & 1.8\tiny$\pm 1.1$ & 0.7\tiny$\pm 0.2$ \\
\midrule
tte-usf-10   & 250         & 0.2  & 0.3  & < 0.1 & < 0.1 & 0.5\tiny$\pm 0.0$ & 1.3\tiny$\pm 0.7$ \\
tte-usf-100  & 2500        & 15.3 & 9.1  & < 0.1 & 0.1  & 0.5\tiny$\pm 0.1$ & 1.4\tiny$\pm 0.5$ \\
tte-usf-1k   & \num{2.5e6} & oot  & oot  & 1.0  & 19.5 & 0.4\tiny$\pm 0.1$ & 1.0\tiny$\pm 1.0$ \\
tte-usf-2k   & \num{1e7}   & oot  & oot  & 3.8  & 81.7 & 0.6\tiny$\pm 0.1$ & 4.8\tiny$\pm 8.4$ \\
tte-usf-5k   & \num{6.25e7}& oot  & oot  & 20.3 & 421.9& 0.6\tiny$\pm 0.1$ & 0.9\tiny$\pm 0.3$ \\
tte-usf-10k  & \num{2.5e8} & oot  & oot  & 83.3 & oot  & 0.7\tiny$\pm 0.1$ & 1.5\tiny$\pm 0.8$ \\
\midrule
con-sf-10    & 250         & 0.2  & 0.3  & < 0.1 & < 0.1 & 1.1\tiny$\pm 0.6$ & 1.5\tiny$\pm 0.4$ \\
con-sf-100   & 2500        & 14.2 & oot  & < 0.1 & < 0.1 & 1.3\tiny$\pm 1.1$ & 1.4\tiny$\pm 0.6$ \\
con-sf-1k    & \num{2.5e6} & oot  & oot  & 1.4  & 13.3 & 1.6\tiny$\pm 1.0$ & 1.6\tiny$\pm 1.0$ \\
con-sf-2k    & \num{1e7}   & oot  & oot  & 5.6  & 53.9 & 1.9\tiny$\pm 1.4$ & 1.1\tiny$\pm 0.7$ \\
con-sf-5k    & \num{6.25e7}& oot  & oot  & 32.4 & oot  & 1.4\tiny$\pm 1.0$ & 1.3\tiny$\pm 0.5$ \\
con-sf-10k   & \num{2.5e8} & oot  & oot  & 137.1& oot  & 2.2\tiny$\pm 1.1$ & 1.4\tiny$\pm 0.7$ \\
\midrule
con-usf-10   & 250         & 0.3  & 0.3  & < 0.1 & < 0.1 & 0.5\tiny$\pm 0.1$ & 1.5\tiny$\pm 0.9$ \\
con-usf-100  & 2500        & 15.1 & 22.9 & < 0.1 & < 0.1 & 0.4\tiny$\pm 0.1$ & 1.5\tiny$\pm 0.8$ \\
con-usf-1k   & \num{2.5e6} & oot  & oot  & 2.5  & 10.1 & 0.3\tiny$\pm 0.1$ & 5.3\tiny$\pm 11.4$ \\
con-usf-2k   & \num{1e7}   & oot  & oot  & 12.1 & 39.2 & 0.4\tiny$\pm 0.1$ & 1.3\tiny$\pm 1.1$ \\
con-usf-5k   & \num{6.25e7}& oot  & oot  & 59.6 & 166.0& 0.4\tiny$\pm 0.1$ & 1.7\tiny$\pm 1.1$ \\
con-usf-10k  & \num{2.5e8} & oot  & oot  & 161.0& oot  & 0.4\tiny$\pm 0.1$ & 1.6\tiny$\pm 1.6$ \\
\bottomrule
\end{tabular}
}
\label{tab:clock}
\vspace{-8mm}
\end{table}

\subsection{Deterministic Infinite-state Termination Analysis}
Table~\ref{tab:term} presents results for the deterministic infinite-state termination benchmarks presented in Figure~\ref{fig:term}. We evaluate benchmarks from program termination analysis, including programs from the SV-COMP termination category~\cite{DBLP:conf/tacas/Beyer23}. These programs operate on unbounded integer variables and may either terminate or enter a non-terminating loop depending on the input. We selected a representative subset of the benchmarks that is currently supported by our implementation, i.e., without data-structures and external library calls. We also had to exclude programs with unbounded non-deterministic variable updates leading to transition systems with unbounded branching.

The baseline tools verify termination for all inputs, while our approach provides the conditions for termination, precisely separating terminating and non-terminating inputs. To run the baseline tools we distinguish between these cases, each program is tested in two versions: one allowing only terminating inputs (“$\text{term}$”) and another including potentially non-terminating inputs (“$\neg\text{term}$”).

\vspace{-10pt}
\begin{table}[ht]
\centering
\caption{Results for deterministic infinite-state termination benchmarks. A ``-'' indicates that there is no such special case for the benchmark.}
\vspace{10pt}
\resizebox{\linewidth}{!}{
\begin{tabular}{c M M M M M M D N}
\toprule
Benchmark      & \multicolumn{2}{c}{\textbf{nuXmv} (IC3)} & \multicolumn{2}{c}{\textbf{CPAChecker}} & \multicolumn{2}{c}{\textbf{Ultimate}} & \multicolumn{2}{c}{\textbf{Bisimulation Learning}} \\
& term        & $\neg$term    & term        & $\neg$term    & term        & $\neg$term    & Deterministic & Branching \\
\midrule
term-loop-1    & oot         & oot           & 0.46        & 1.12          & 0.92        & 1.08          & 0.8\tiny$\pm$0.7                         & 0.2\tiny$\pm$0.0          \\
term-loop-2    & oot         & oot           & 0.46        & 0.21          & 0.48        & 0.11          & 0.5\tiny$\pm$0.2                         & 0.3\tiny$\pm$0.0          \\
audio-compr    & $<0.1$      & $<0.1$        & 4.79        & 1.87          & 0.51        & 0.37          & 0.4\tiny$\pm$0.1                         & 0.3\tiny$\pm$0.1          \\
euclid         & oot         & oot           & 0.54        & 0.23          & 0.87        & 0.21          & 1.2\tiny$\pm$0.2                         & 1.0\tiny$\pm$0.6          \\
greater        & 9.3         & $<0.1$        & 0.46        & 0.21          & 0.56        & 0.20          & 1.0\tiny$\pm$0.8                         & 0.5\tiny$\pm$0.1          \\
smaller        & 9.3         & $<0.1$        & 0.46        & 0.22          & 0.65        & 0.20          & 0.6\tiny$\pm$0.1                         & 0.5\tiny$\pm$0.1          \\
conic          & 401.5       & $<0.1$        & 0.49        & 0.23          & 3.51        & 0.20          & 2.1\tiny$\pm$0.6                         & 1.8\tiny$\pm$0.7          \\
disjunction    & oot         & -           & 24.53       & -             & 0.69        & -             & 0.4\tiny$\pm$0.1                         & 0.3\tiny$\pm$0.1          \\
parallel       & oot         & -             & 0.63        & -             & 1.60        & -             & 0.4\tiny$\pm$0.1                         & 0.4\tiny$\pm$0.1          \\
quadratic      & oot         & -             & 0.21        & -             & 1.28        & -             & 0.8\tiny$\pm$0.2                         & 0.7\tiny$\pm$0.1          \\
cubic          & 0.5         & $<0.1$        & 0.23        & 0.46          & 1.23        & 0.56          & 0.7\tiny$\pm$0.2                         & 0.6\tiny$\pm$0.2          \\
nlr-cond       & 0.5         & $<0.1$        & 0.51        & 0.49          & 1.30        & 0.11          & 0.6\tiny$\pm$0.2                         & 0.6\tiny$\pm$0.1          \\
\bottomrule
\end{tabular}
}
\label{tab:term}
\end{table}

\subsection{Non-deterministic Infinite-State Systems}

In the following, we present the results for non-deterministic infinite-state systems, as depicted in Figure~\ref{fig:ltlctl}. These include non-deterministic versions of the termination benchmarks from the previous section, concurrent programs from the benchmark set used for the T2 verifier~\cite{DBLP:conf/cav/BeyenePR13,DBLP:conf/tacas/BrockschmidtCIK16}, and reactive robotic case studies. From the T2 benchmark set, we include only those benchmarks that could be reliably translated from the available file format.

Table~\ref{tab:ltlinf} presents results for linear-time properties expressed in \(\ltln\), comparing our approach against the nuXmv model checker (IC3) and the Ultimate verifier. Table~\ref{tab:ctlinf} reports results for branching-time properties in \(\ctlnn\) and \(\ctln\), compared against T2, the only available competitor for branching-time verification of non-deterministic infinite-state systems. 
For bisimulation learning, we report a single runtime, as it synthesises a very succinct quotient for all considered benchmarks that enables verification of arbitrary \(\ctln\) properties using a finite-state model checker in negligible time.

\begin{longtable}{c N N N N}
\caption{Updated LTL specs on non-deterministic infinite state program benchmarks.}\\

\toprule
Benchmark & $\varphi$ & \textbf{\small nuXmv} (IC3) & \textbf{\small Ultimate} & \textbf{\small Bisimulation Learning} \\
\midrule
\endfirsthead
\toprule
Benchmark & $\varphi$ & \textbf{\small nuXmv} (IC3) & \textbf{\small Ultimate} & \textbf{\small Bisimulation Learning} \\
\midrule
\endhead
\bottomrule
\endfoot
\multirow[t]{4}{*}{term-loop-nd} 
  & \tiny\texttt{(x < -1 -> G(!terminated))} & <0.1 & 0.12\tiny$\pm$0.04 & 0.42\tiny$\pm$0.12 \\
  & \tiny\texttt{(x > 1 -> F(terminated))}    & <0.1 & 0.14\tiny$\pm$0.05 &  \\
  & \tiny\texttt{F(G(x >= -1))}                & oot  & 3.34\tiny$\pm$0.07 &  \\
  & \tiny\texttt{G(F(x >= -1))}                & oot  & 3.21\tiny$\pm$0.03 &  \\ \midrule
\multirow[t]{4}{*}{term-loop-nd-2} 
  & \tiny\texttt{Ours: x < 0 -> !F(terminated)}  & <0.1 & 0.11\tiny$\pm$0.03 & 1.17\tiny$\pm$0.51 \\
  & \tiny\texttt{x > 0 -> F(terminated)}         & <0.1 & 0.20\tiny$\pm$0.00 &  \\
  & \tiny\texttt{F(G(x > 0))}                     & <0.1 & 0.21\tiny$\pm$0.03 &  \\
  & \tiny\texttt{G(F(x > 0))}                     & <0.1 & 0.20\tiny$\pm$0.00 &  \\ \midrule
\multirow[t]{4}{*}{term-loop-nd-y} 
  & \tiny\texttt{x > 0 -> F(terminated)}         & <0.1 & 0.14\tiny$\pm$0.05 & 1.14\tiny$\pm$0.30 \\
  & \tiny\texttt{(x < y -> F(terminated))}         & oot  & 0.19\tiny$\pm$0.03 &  \\
  & \tiny\texttt{x < y -> G F(terminated)}         & oot  & 0.14\tiny$\pm$0.05 &  \\
  & \tiny\texttt{G F(terminated)}                  & <0.1 & 0.71\tiny$\pm$0.05 &  \\ \midrule
\multirow[t]{3}{*}{quadratic-nd} 
  & \tiny\texttt{x > 1 -> F(terminated)}         & oot  & 0.11\tiny$\pm$0.03 & 1.05\tiny$\pm$0.35 \\
  & \tiny\texttt{F(terminated)}                  & <0.1 & 0.91\tiny$\pm$0.07 &  \\
  & \tiny\texttt{x > 1 -> F G(terminated)}         & oot  & 0.12\tiny$\pm$0.04 &  \\ \midrule
\multirow[t]{3}{*}{cubic-nd} 
  & \tiny\texttt{x > 1 -> F(terminated)}         & oot  & 0.10\tiny$\pm$0.00 & 0.72\tiny$\pm$0.29 \\
  & \tiny\texttt{F(terminated)}                  & oot  & 0.70\tiny$\pm$0.08 &  \\
  & \tiny\texttt{x > 1 -> F G(terminated)}         & oot  & 0.11\tiny$\pm$0.03 &  \\ \midrule
\multirow[t]{3}{*}{nlr-cond-nd} 
  & \tiny\texttt{G(!terminated)}                & 0.166 & 0.50\tiny$\pm$0.00 & 1.16\tiny$\pm$0.24 \\
  & \tiny\texttt{F(terminated)}                  & 0.115 & 4.34\tiny$\pm$0.05 &  \\
  & \tiny\texttt{x > 1 -> F G(terminated)}         & oot   & 0.15\tiny$\pm$0.05 &  \\ \midrule
\multirow[t]{4}{*}{P1} 
  & \tiny\texttt{G (n >= 0)}                     & <0.1 & 0.56\tiny$\pm$0.05 & 1.09\tiny$\pm$0.34 \\
  & \tiny\texttt{F (a = 1)}                      & <0.1 & 0.65\tiny$\pm$0.05 &  \\
  & \tiny\texttt{G F(n > 0)}                     & 3.178 & 3.1\tiny$\pm$0.15  &  \\
  & \tiny\texttt{r < a -> F G(r < a)}             & 4.121 & 0.21\tiny$\pm$0.03 &  \\ \midrule
\multirow[t]{4}{*}{P2} 
  & \tiny\texttt{G (n >= 0)}                     & <0.1 & 0.53\tiny$\pm$0.06 & 1.01\tiny$\pm$0.22 \\
  & \tiny\texttt{F (a = 1)}                      & <0.1 & 0.65\tiny$\pm$0.05 &  \\
  & \tiny\texttt{G F(n > 0)}                     & 0.484 & 3.0\tiny$\pm$0.13  &  \\
  & \tiny\texttt{r < a -> F G(r < a)}             & 1.134 & 0.2\tiny$\pm$0.0   &  \\ \midrule
\multirow[t]{4}{*}{P3} 
  & \tiny\texttt{G(a = 1 -> F(r = 1))}          & <0.1 & 0.24\tiny$\pm$0.00 & 0.30\tiny$\pm$0.00 \\
  & \tiny\texttt{F (a = 1)}                      & <0.1 & 0.64\tiny$\pm$0.05 &  \\
  & \tiny\texttt{G F(n > 0)}                     & <0.1 & 0.2\tiny$\pm$0.0   &  \\
  & \tiny\texttt{r < a -> F G(r < a)}             & <0.1 & 0.2\tiny$\pm$0.0   &  \\ \midrule
\multirow[t]{4}{*}{P4} 
  & \tiny\texttt{G(a = 1 -> F(r = 1))}          & <0.1 & 1.31\tiny$\pm$0.17 & oot \\
  & \tiny\texttt{F (n = 1)}                      & <0.1 & 0.6\tiny$\pm$0.04  &  \\
  & \tiny\texttt{G F(n > 0)}                     & oot  & 34.7\tiny$\pm$0.51 &  \\
  & \tiny\texttt{r < a -> F G(r < a)}             & oot  & 0.2\tiny$\pm$0.0   &  \\ \midrule
\multirow[t]{4}{*}{P5} 
  & \tiny\texttt{G(s = 1 -> F(u = 1))}          & <0.1 & 0.81\tiny$\pm$0.11 & 1.78\tiny$\pm$0.35 \\
  & \tiny\texttt{F(s = 1 \& u = 1)}            & <0.1 & 1.92\tiny$\pm$0.04 &  \\
  & \tiny\texttt{G F (u = 1)}                   & <0.1 & 1.19\tiny$\pm$0.05 &  \\
  & \tiny\texttt{p < i -> F G(u = 1)}            & <0.1 & 0.66\tiny$\pm$0.05 &  \\ \midrule
\multirow[t]{4}{*}{P6} 
  & \tiny\texttt{G(s = 1 -> F(u = 1))}          & <0.1 & 0.97\tiny$\pm$0.06 & 0.36\tiny$\pm$0.01 \\
  & \tiny\texttt{F(s = 1 \& u = 1)}            & <0.1 & 1.01\tiny$\pm$0.05 &  \\
  & \tiny\texttt{G F (u = 1)}                   & <0.1 & 0.88\tiny$\pm$0.06 &  \\
  & \tiny\texttt{r < a -> F G(r < a)}            & <0.1 & 0.5\tiny$\pm$0.0   &  \\ \midrule
\multirow[t]{4}{*}{P7} 
  & \tiny\texttt{G(s = 1 -> F(u = 1))}          & <0.1 & 0.97\tiny$\pm$0.06 & 0.36\tiny$\pm$0.00 \\
  & \tiny\texttt{F(s = 1 \& u = 1)}            & <0.1 & 1.0\tiny$\pm$0.04  &  \\
  & \tiny\texttt{G F (u = 1)}                   & <0.1 & 0.92\tiny$\pm$0.04 &  \\
  & \tiny\texttt{p < i -> F G(u = 1)}            & <0.1 & 0.5\tiny$\pm$0.0   &  \\ \midrule
\multirow[t]{4}{*}{P17} 
  & \tiny\texttt{x <= 1 -> G(x <= 1)}           & <0.1 & 2.92\tiny$\pm$0.07 & 0.53\tiny$\pm$0.09 \\
  & \tiny\texttt{F(x>= 1)}                     & oot  & 6.44\tiny$\pm$0.22 &  \\
  & \tiny\texttt{G(F(x >= 1))}                  & oot  & 0.48\tiny$\pm$0.07 &  \\
  & \tiny\texttt{F G (x >= 1)}                  & oot  & oot               &  \\ \midrule
\multirow[t]{4}{*}{P19} 
  & \tiny\texttt{x <= 1 -> G(x <= 1)}           & <0.1 & 3.89\tiny$\pm$0.05 & 0.54\tiny$\pm$0.12 \\
  & \tiny\texttt{F(x>= 1)}                     & oot  & 5.37\tiny$\pm$0.08 &  \\
  & \tiny\texttt{G(F(x >= 1))}                  & oot  & 0.59\tiny$\pm$0.03 &  \\
  & \tiny\texttt{F G (x >= 1)}                  & oot  & oot               &  \\ \midrule
\multirow[t]{4}{*}{P20} 
  & \tiny\texttt{x <= 1 -> G(x <= 1)}           & <0.1 & 0.0               & 0.67\tiny$\pm$0.10 \\
  & \tiny\texttt{F(x < 1)}                     & oot  & 0.15\tiny$\pm$0.05 &  \\
  & \tiny\texttt{G(F(x < 1))}                  & oot  & 0.18\tiny$\pm$0.04 &  \\
  & \tiny\texttt{F G (x < 1)}                  & oot  & 0.0               &  \\ \midrule
\multirow[t]{2}{*}{P21} 
  & \tiny\texttt{G F w != 1}                   & <0.1 & 0.27\tiny$\pm$0.05 & 0.37\tiny$\pm$0.00 \\
  & \tiny\texttt{G F w = 5}                    & <0.1 & 0.71\tiny$\pm$0.03 &  \\ \midrule
\multirow[t]{4}{*}{P25} 
  & \tiny\texttt{G c > 5}                      & oot  & 0.34\tiny$\pm$0.05 & 0.31\tiny$\pm$0.00 \\
  & \tiny\texttt{F c < 5}                      & <0.1 & 0.2\tiny$\pm$0.0   &  \\
  & \tiny\texttt{G F r >= 5}                   & <0.1 & 0.98\tiny$\pm$0.06 &  \\
  & \tiny\texttt{F G c > 5}                    & <0.1 & 0.39\tiny$\pm$0.03 &  \\ \midrule
\multirow[t]{3}{*}{two-robots} 
  & \tiny\texttt{x\_1 >= 1 -> G(!clash)}       & 0.02 & 0.38\tiny$\pm$0.04 & 0.52\tiny$\pm$0.16 \\
  & \tiny\texttt{y\_2 >= 1 -> G(!clash)}       & 0.02 & 0.13\tiny$\pm$0.05 &  \\
  & \tiny\texttt{F ((G clash) | (G !clash))}     & 0.03 & 28.23\tiny$\pm$0.11 &  \\ \midrule
\multirow[t]{3}{*}{two-robots-actions} 
  & \tiny\texttt{x\_1 >= 1 -> G(!clash)}       & 0.02 & 0.44\tiny$\pm$0.05 & 1.0\tiny$\pm$0.2 \\
  & \tiny\texttt{y\_2 >= 1 -> G(!clash)}       & 0.02 & 0.17\tiny$\pm$0.05 &  \\
  & \tiny\texttt{F ((G clash) | (G !clash))}     & 0.03 & 16.22\tiny$\pm$0.19 &  \\ \midrule
\multirow[t]{3}{*}{two-robots-quadratic} 
  & \tiny\texttt{x\_1 >= 1 -> G(!clash)}       & 0.02 & 0.15\tiny$\pm$0.05 & 0.9\tiny$\pm$0.27 \\
  & \tiny\texttt{y\_2 >= 1 -> G(!clash)}       & 0.02 & 0.15\tiny$\pm$0.05 &  \\
  & \tiny\texttt{F ((G clash) | (G !clash))}     & 0.03 & 5.89\tiny$\pm$0.07  & 
  \label{tab:ltlinf}
\end{longtable}

\newcolumntype{E}{>{\centering\arraybackslash}m{2.6cm}}
\newcolumntype{F}{>{\centering\arraybackslash}m{3cm}}
\begin{longtable}{E F F N}
\caption{Performance Comparison of T2 and Branching Bisimulation Learning on Non-Deterministic Infinite-State Program Benchmarks.}\\[1mm]
\toprule
Benchmark & $\phi$ & \textbf{T2} & \textbf{\small Bisimulation Learning} \\
\midrule
\endfirsthead
\toprule
Benchmark & $\phi$ & \textbf{T2} & \textbf{\small Bisimulation Learning} \\
\midrule
\endhead
term-loop-nd 
    & \tiny\texttt{G(varX > 1 || varX < -1)} & 0.30\tiny$\pm$0.00 & 0.25\tiny$\pm$0.04 \\
    & \tiny\texttt{F(G(varX <= 1) || G(varX >= -1))} & 5.33\tiny$\pm$0.68 & \\
    & \tiny\texttt{G(F(varX > 1))} & oot & \\
    & \tiny\texttt{[EG](varX > 1 || varX < -1)} & 0.62\tiny$\pm$0.00 & \\
    & \tiny\texttt{!([EG](varX > 1 || varX < -1))} & 0.62\tiny$\pm$0.00 & \\
    & \tiny\texttt{[AF]([AG](varX <= 1) \&\& [AG](varX >= -1))} & 0.59\tiny$\pm$0.00 & \\
    & \tiny\texttt{A F(G(F(varX <= 1)) || G(F(varX >= -1)))} & oot & \\
\midrule
term-loop-nd-2 
    & \tiny\texttt{G(F(varX == 0))} & 1.89\tiny$\pm$0.00 & 0.78\tiny$\pm$0.32 \\
    & \tiny\texttt{F(G(varX <= 1) || G(varX >= -1))} & oot & \\
    & \tiny\texttt{[EG]([AF](varX == 0))} & 1.16\tiny$\pm$0.13 & \\
    & \tiny\texttt{!([EG]([AF](varX == 0)))} & oot & \\
    & \tiny\texttt{[AF]([AG](varX <= 1) \&\& [AG](varX >= -1))} & 1.64\tiny$\pm$0.01 & \\
    & \tiny\texttt{A F(E G( F(varX <= 1)) \&\& A G(F(varX >= -1)))} & oot & \\
\midrule
term-loop-nd-y 
    & \tiny\texttt{[EG]([AF](varX == 0))} & 1.09\tiny$\pm$0.01 & oot \\
    & \tiny\texttt{!([EG]([AF](varX == 0)))} & 1.12\tiny$\pm$0.00 & \\
    & \tiny\texttt{G(F(varX == 0))} & 1.50\tiny$\pm$0.01 & \\
    & \tiny\texttt{E G(F(varX == 0) \&\& F(G(varX != 0)))} & 7.86\tiny$\pm$0.27 & \\
\midrule
P1 
    & \tiny\texttt{[AG](varA != 1 || [AF](varR == 1))} & 0.68\tiny$\pm$0.00 & 0.56\tiny$\pm$0.05 \\
    & \tiny\texttt{!([AG](varA != 1 || [AF](varR == 1)))} & 0.68\tiny$\pm$0.00 & \\
    & \tiny\texttt{G(varA != 1 || F(varR == 1))} & 1.04\tiny$\pm$0.05 & \\
    & \tiny\texttt{E G(F(varA != 1 || F(varR == 1)))} & 2.33\tiny$\pm$0.06 & \\
\midrule
P2 
    & \tiny\texttt{[EF](varA == 1 \&\& [EG](varR != 5))} & 0.82\tiny$\pm$0.00 & 0.53\tiny$\pm$0.06 \\
    & \tiny\texttt{!([EF](varA == 1 \&\& [EG](varR != 5)))} & 0.81\tiny$\pm$0.01 & \\
    & \tiny\texttt{F(varA == 1 \&\& G(varR != 5))} & 1.06\tiny$\pm$0.01 & \\
    & \tiny\texttt{E G (F(varA == 1 \&\& E G(varR != 5)))} & oot & \\
\midrule
P3 
    & \tiny\texttt{[AG](varA != 1 || [EF](varR == 1))} & 0.37\tiny$\pm$0.00 & 0.24\tiny$\pm$0.00 \\
    & \tiny\texttt{!([AG](varA != 1 || [EF](varR == 1)))} & 0.37\tiny$\pm$0.00 & \\
    & \tiny\texttt{G(varA != 1 || F(varR == 1))} & 0.97\tiny$\pm$0.00 & \\
    & \tiny\texttt{E F( G(varA != 1 || A G(varR == 1)))} & 8.58\tiny$\pm$0.04 & \\
\midrule
P4 
    & \tiny\texttt{[EF](varA == 1 \&\& [AG](varR != 1))} & oot & 1.31\tiny$\pm$0.37 \\
    & \tiny\texttt{![EF](varA == 1 \&\& [AG](varR != 1))} & oot & \\
    & \tiny\texttt{F(varA == 1 \&\& G(varR != 1))} & 0.58\tiny$\pm$0.00 & \\
    & \tiny\texttt{E G (F(varA == 1 \&\& E G(varR != 1)))} & 2.02\tiny$\pm$0.05 & \\
\midrule
P5 
    & \tiny\texttt{[AG](varS != 1 || [AF](varU == 1))} & 0.53\tiny$\pm$0.00 & 0.81\tiny$\pm$0.11 \\
    & \tiny\texttt{!([AG](varS != 1 || [AF](varU == 1)))} & 0.52\tiny$\pm$0.00 & \\
    & \tiny\texttt{G(varS != 1 || F(varU == 1))} & oot & \\
    & \tiny\texttt{E F (G(varS != 1 || A F(varU == 1)))} & 0.42\tiny$\pm$0.00 & \\
\midrule
P6 
    & \tiny\texttt{[EF](varS == 1 || [EG](varU != 1))} & 0.53\tiny$\pm$0.00 & 0.28\tiny$\pm$0.00 \\
    & \tiny\texttt{!([EF](varS == 1 || [EG](varU != 1)))} & 0.53\tiny$\pm$0.00 & \\
    & \tiny\texttt{F(varS == 1 || G(varU != 1))} & 1.25\tiny$\pm$0.01 & \\
    & \tiny\texttt{E G (F(varS == 1 \&\& E F(varU != 1)))} & 0.61\tiny$\pm$0.00 & \\
\midrule
P7 
    & \tiny\texttt{[AG](varS != 1 || [EF](varU == 1))} & 0.47\tiny$\pm$0.00 & 0.28\tiny$\pm$0.00 \\
    & \tiny\texttt{!([AG](varS != 1 || [EF](varU == 1)))} & 0.45\tiny$\pm$0.00 & \\
    & \tiny\texttt{G(varS != 1 || F(varU == 1))} & 0.46\tiny$\pm$0.00 & \\
    & \tiny\texttt{E G(varS != 1 || F(varU == 1) || E F(varU == 1))} & 1.25\tiny$\pm$0.10 & \\
\midrule
P17 
    & \tiny\texttt{[AG]([AF](varW >= 1))} & 0.58\tiny$\pm$0.00 & 0.36\tiny$\pm$0.06 \\
    & \tiny\texttt{!([AG]([AF](varW >= 1)))} & 0.58\tiny$\pm$0.00 & \\
    & \tiny\texttt{G(F(varW >= 1))} & 0.85\tiny$\pm$0.00 & \\
    & \tiny\texttt{E G( F(varW >= 1) || A G (E F(varW >= 1)))} & 2.25\tiny$\pm$0.03 & \\
\midrule
P18 
    & \tiny\texttt{[EF]([EG](varW < 1))} & 0.75\tiny$\pm$0.00 & 0.28\tiny$\pm$0.03 \\
    & \tiny\texttt{!([EF]([EG](varW < 1)))} & 0.75\tiny$\pm$0.00 & \\
    & \tiny\texttt{F(G(varW < 1))} & 3.57\tiny$\pm$0.01 & \\
    & \tiny\texttt{E F(F(varW >= 1)) || A F (E G(varW < 1))} & 2.54\tiny$\pm$0.02 & \\
\midrule
P19 
    & \tiny\texttt{[AG]([EF](varW >=1))} & 3.64\tiny$\pm$0.01 & 0.33\tiny$\pm$0.03 \\
    & \tiny\texttt{!([AG]([EF](varW >=1)))} & 3.65\tiny$\pm$0.01 & \\
    & \tiny\texttt{G(F(varW >=1))} & oot & \\
    & \tiny\texttt{E G(F(varW >=1) || E F (A G(varW >=1)))} & oot & \\
\midrule
P20 
    & \tiny\texttt{[EF]([AG](varW < 1))} & 0.73\tiny$\pm$0.00 & 0.36\tiny$\pm$0.04 \\
    & \tiny\texttt{!([EF]([AG](varW < 1)))} & 0.88\tiny$\pm$0.00 & \\
    & \tiny\texttt{F(G(varW < 1))} & oot & \\
    & \tiny\texttt{E F(G(varW >= 1)) || A G(F(varW >= 1))} & 71.80\tiny$\pm$0.59 & \\
\midrule
P21 
    & \tiny\texttt{[AG]([AF](varW == 1))} & 0.43\tiny$\pm$0.11 & 0.31\tiny$\pm$0.00 \\
    & \tiny\texttt{!([AG]([AF](varW == 1)))} & 0.39\tiny$\pm$0.00 & \\
    & \tiny\texttt{G(F(varW == 1))} & 0.44\tiny$\pm$0.01 & \\
    & \tiny\texttt{E F(G(varW >= 1)) || A G(F(varW >= 1))} & oot & \\
\midrule
P22 
    & \tiny\texttt{[EF]([EG](varW != 1))} & 0.78\tiny$\pm$0.01 & 0.33\tiny$\pm$0.00 \\
    & \tiny\texttt{!([EF]([EG](varW != 1)))} & 0.78\tiny$\pm$0.01 & \\
    & \tiny\texttt{F(G(varW != 1))} & 0.73\tiny$\pm$0.01 & \\
    & \tiny\texttt{E F(G(varW >= 1)) || A G(F(varW >= 1))} & oot & \\
\midrule
P23 
    & \tiny\texttt{[AG]([EF](varW == 1))} & 0.43\tiny$\pm$0.00 & 0.68\tiny$\pm$0.09 \\
    & \tiny\texttt{!([AG]([EF](varW == 1)))} & 0.42\tiny$\pm$0.00 & \\
    & \tiny\texttt{G(F(varW == 1))} & 3.14\tiny$\pm$0.01 & \\
    & \tiny\texttt{E F(G(varW >= 1)) || A G(F(varW >= 1))} & 3.12\tiny$\pm$0.05 & \\
\midrule
P24 
    & \tiny\texttt{[EF]([AG](varW != 1))} & 0.41\tiny$\pm$0.00 & 0.08\tiny$\pm$0.00 \\
    & \tiny\texttt{!([EF]([AG](varW != 1)))} & 0.42\tiny$\pm$0.00 & \\
    & \tiny\texttt{F(G(varW != 1))} & 0.96\tiny$\pm$0.00 & \\
    & \tiny\texttt{E F(G(varW >= 1)) || A G(F(varW >= 1))} & 12.05\tiny$\pm$0.16 & \\
\midrule
P25 
    & \tiny\texttt{(varC > 5) -> ([AF](varR > 5))} & 0.18\tiny$\pm$0.00 & 0.24\tiny$\pm$0.00 \\
    & \tiny\texttt{(varC > 5) \&\& ![AF](varR > 5)} & oot & \\
    & \tiny\texttt{(varC > 5) -> (F(varR > 5))} & 0.22\tiny$\pm$0.00 & \\
    & \tiny\texttt{E G (F(varC == 5) || G(varC > 5) \&\& F(G(varR > 5)))} & 2.92\tiny$\pm$0.30 & \\
\midrule
P26 
    & \tiny\texttt{(varC > 5) \&\& [EG](varR <= 5)} & 0.52\tiny$\pm$0.00 & 0.24\tiny$\pm$0.00 \\
    & \tiny\texttt{!((varC > 5) \&\& [EG](varR <= 5))} & 0.52\tiny$\pm$0.00 & \\
    & \tiny\texttt{(varC > 5) \&\& G(varR <= 5)} & 0.54\tiny$\pm$0.00 & \\
    & \tiny\texttt{E G (F(varC == 5) || G(varC > 5) \&\& F(G(varR > 5)))} & 1.28\tiny$\pm$0.01 & \\
\midrule
P27 
    & \tiny\texttt{(varC <= 5) || [EF](varR > 5)} & 0.64\tiny$\pm$0.00 & 0.24\tiny$\pm$0.00 \\
    & \tiny\texttt{!((varC <= 5) || [EF](varR > 5))} & 0.65\tiny$\pm$0.00 & \\
    & \tiny\texttt{(varC <= 5) || F(varR > 5)} & 0.67\tiny$\pm$0.14 & \\
    & \tiny\texttt{E G (F(varC == 5) || G(varC > 5) \&\& F(G(varR > 5)))} & oot & \\
\midrule
P28 
    & \tiny\texttt{(varC > 5) \&\& [AG](varR <= 5)} & 0.43\tiny$\pm$0.00 & 0.21\tiny$\pm$0.00 \\
    & \tiny\texttt{!((varC > 5) \&\& [AG](varR <= 5))} & 0.42\tiny$\pm$0.00 & \\
    & \tiny\texttt{(varC > 5) \&\& G(varR <= 5)} & 0.48\tiny$\pm$0.00 & \\
    & \tiny\texttt{E G (F(varC == 5) || G(varC > 5) \&\& F(G(varR > 5)))} & oot & \\
\bottomrule
\label{tab:ctlinf}
\end{longtable}

\subsection{Non-deterministic Finite-state Benchmarks}

Table~\ref{tab:ctlfin} presents results for finite-state versions of the non-deterministic benchmarks from Tables~\ref{tab:ltlinf} and~\ref{tab:ctlinf}. We evaluate multiple state-space sizes determined by variable ranges. As these benchmarks are finite-state, we compare against the nuXmv model checker using symbolic model checking with BDDs. We only report a single runtime for branching bisimulation learning as we synthesise the quotient for the infinite-state system with unbounded variables, subsuming all finite-state versions.

\begin{longtable}{N N E N N}
\caption{Finite state CTL evaluations.}\\[1mm]
\toprule
Benchmark & $|S|$ & $\varphi$ & \textbf{\small nuXmv} \small (BDD) & \textbf{\small Bisimulation Learning} \\
\midrule
\endfirsthead
\toprule
Benchmark & $|S|$ & $\varphi$ & \textbf{\small nuXmv} \small (BDD) & \textbf{\small Bisimulation Learning} \\
\midrule
\endhead
\multirow[t]{15}{*}{term-loop-nd} 
    & $2^{9}$  & {\tiny\texttt{EG(x > 1 | x < -1)}} & 0.02\tiny$\pm$0.00  & 0.42\tiny$\pm$0.12 \\[1mm]
    & $2^{11}$ &                                      & 0.32\tiny$\pm$0.00  &  \\[1mm]
    & $2^{13}$ &                                      & 4.40\tiny$\pm$0.03  &  \\[1mm]
    & $2^{15}$ &                                      & 76.66\tiny$\pm$4.08  &  \\[1mm]
    & $2^{17}$ &                                      & oot                &  \\[1mm]
    & $2^{9}$  & {\tiny\texttt{!(EG(x > 1 | x < -1))}}& 0.02\tiny$\pm$0.00  &  \\[1mm]
    & $2^{11}$ &                                      & 0.32\tiny$\pm$0.00  &  \\[1mm]
    & $2^{13}$ &                                      & 4.55\tiny$\pm$0.06  &  \\[1mm]
    & $2^{15}$ &                                      & 77.25\tiny$\pm$0.54  &  \\[1mm]
    & $2^{17}$ &                                      & oot                &  \\[1mm]
    & $2^{9}$  & {\tiny\texttt{AF(AG(x $\le$ 1) \& AG(x $\ge$ -1))}} & 0.02\tiny$\pm$0.00  &  \\[1mm]
    & $2^{11}$ &                                      & 0.32\tiny$\pm$0.00  &  \\[1mm]
    & $2^{13}$ &                                      & 4.32\tiny$\pm$0.04  &  \\[1mm]
    & $2^{15}$ &                                      & 73.92\tiny$\pm$0.69  &  \\[1mm]
    & $2^{17}$ &                                      & oot                &  \\
\midrule
\multirow[t]{15}{*}{term-loop-nd-2} 
    & $2^{9}$  & {\tiny\texttt{EG(AF(x = 0))}}          & 0.03\tiny$\pm$0.00  & 1.17\tiny$\pm$0.51 \\[1mm]
    & $2^{11}$ &                                      & 0.67\tiny$\pm$0.01  &  \\[1mm]
    & $2^{13}$ &                                      & 10.28\tiny$\pm$0.06 &  \\[1mm]
    & $2^{15}$ &                                      & 176.68\tiny$\pm$10.54&  \\[1mm]
    & $2^{17}$ &                                      & oot                &  \\[1mm]
    & $2^{9}$  & {\tiny\texttt{!(EG(AF(x = 0)))} }       & 0.03\tiny$\pm$0.00  &  \\[1mm]
    & $2^{11}$ &                                      & 0.68\tiny$\pm$0.01  &  \\[1mm]
    & $2^{13}$ &                                      & 10.34\tiny$\pm$0.07 &  \\[1mm]
    & $2^{15}$ &                                      & 175.41\tiny$\pm$4.20 &  \\[1mm]
    & $2^{17}$ &                                      & oot                &  \\[1mm]
    & $2^{9}$  & {\tiny\texttt{AF(AG(x $\le$ 1) \& AG(x $\ge$ -1))}} & 0.03\tiny$\pm$0.00  &  \\[1mm]
    & $2^{11}$ &                                      & 0.68\tiny$\pm$0.01  &  \\[1mm]
    & $2^{13}$ &                                      & 9.90\tiny$\pm$0.07  &  \\[1mm]
    & $2^{15}$ &                                      & 168.72\tiny$\pm$1.14 &  \\[1mm]
    & $2^{17}$ &                                      & oot                &  \\
\midrule
\multirow[t]{15}{*}{term-loop-nd-y} 
    & $2^{9}$  & {\tiny\texttt{EG(AF(x = 0))}}          & 0.05\tiny$\pm$0.00  & 1.14\tiny$\pm$0.30 \\[1mm]
    & $2^{11}$ &                                      & 0.64\tiny$\pm$0.01  &  \\[1mm]
    & $2^{13}$ &                                      & 12.55\tiny$\pm$3.43 &  \\[1mm]
    & $2^{15}$ &                                      & 153.32\tiny$\pm$0.58&  \\[1mm]
    & $2^{17}$ &                                      & oot                &  \\[1mm]
    & $2^{9}$  & {\tiny\texttt{!(EG(AF(x = 0)))} }       & 0.08\tiny$\pm$0.00  &  \\[1mm]
    & $2^{11}$ &                                      & 1.19\tiny$\pm$0.01  &  \\[1mm]
    & $2^{13}$ &                                      & 23.12\tiny$\pm$5.23 &  \\[1mm]
    & $2^{15}$ &                                      & 353.75\tiny$\pm$3.38&  \\[1mm]
    & $2^{17}$ &                                      & oot                &  \\[1mm]
    & $2^{9}$  & {\tiny\texttt{AF(AG(x $\le$ 1) \& AG(x $\ge$ -1))}} & 0.05\tiny$\pm$0.00  &  \\[1mm]
    & $2^{11}$ &                                      & 0.61\tiny$\pm$0.01  &  \\[1mm]
    & $2^{13}$ &                                      & 11.90\tiny$\pm$0.08 &  \\[1mm]
    & $2^{15}$ &                                      & 153.13\tiny$\pm$3.16 &  \\[1mm]
    & $2^{17}$ &                                      & oot                &  \\
\midrule
\multirow[t]{10}{*}{P1} 
    & $2^{9}$  & {\tiny\texttt{AG(a $\neq$ 1 | AF(r = 1))}} & 0.04\tiny$\pm$0.00  & 1.09\tiny$\pm$0.34 \\[1mm]
    & $2^{11}$ &                                      & 0.51\tiny$\pm$0.01  &  \\[1mm]
    & $2^{13}$ &                                      & 7.38\tiny$\pm$0.06  &  \\[1mm]
    & $2^{15}$ &                                      & 125.00\tiny$\pm$1.08 &  \\[1mm]
    & $2^{17}$ &                                      & oot                &  \\[1mm]
    & $2^{9}$  & {\tiny\texttt{!(AG(a $\neq$ 1 | AF(r = 1)))} } & 0.04\tiny$\pm$0.00  &  \\[1mm]
    & $2^{11}$ &                                      & 0.50\tiny$\pm$0.02  &  \\[1mm]
    & $2^{13}$ &                                      & 9.04\tiny$\pm$2.59  &  \\[1mm]
    & $2^{15}$ &                                      & 171.09\tiny$\pm$54.25 &  \\[1mm]
    & $2^{17}$ &                                      & oot                &  \\
\midrule
\multirow[t]{10}{*}{P2} 
    & $2^{9}$  & {\tiny\texttt{EF(a = 1 \& EG(r $\neq$ 5))}} & 0.04\tiny$\pm$0.00  & 1.01\tiny$\pm$0.22 \\[1mm]
    & $2^{11}$ &                                      & 0.51\tiny$\pm$0.01  &  \\[1mm]
    & $2^{13}$ &                                      & 13.89\tiny$\pm$0.10 &  \\[1mm]
    & $2^{15}$ &                                      & 174.70\tiny$\pm$54.21 &  \\[1mm]
    & $2^{17}$ &                                      & oot                &  \\[1mm]
    & $2^{9}$  & {\tiny\texttt{!EF(a = 1 \& EG(r $\neq$ 5))}} & 0.04\tiny$\pm$0.00  &  \\[1mm]
    & $2^{11}$ &                                      & 0.49\tiny$\pm$0.02  &  \\[1mm]
    & $2^{13}$ &                                      & 13.69\tiny$\pm$0.07 &  \\[1mm]
    & $2^{15}$ &                                      & 158.67\tiny$\pm$49.57 &  \\[1mm]
    & $2^{17}$ &                                      & oot                &  \\
\midrule
\multirow[t]{10}{*}{P3} 
    & $2^{9}$  & {\tiny\texttt{AG(a $\neq$ 1 | EF(r = 1))}} & 0.04\tiny$\pm$0.00  & 0.30\tiny$\pm$0.00 \\[1mm]
    & $2^{11}$ &                                      & 0.28\tiny$\pm$0.00  &  \\[1mm]
    & $2^{13}$ &                                      & 4.02\tiny$\pm$0.39  &  \\[1mm]
    & $2^{15}$ &                                      & 65.09\tiny$\pm$0.39 &  \\[1mm]
    & $2^{17}$ &                                      & oot                &  \\[1mm]
    & $2^{9}$  & {\tiny\texttt{!(AG(a $\neq$ 1 | EF(r = 1)))} } & 0.03\tiny$\pm$0.00  &  \\[1mm]
    & $2^{11}$ &                                      & 0.29\tiny$\pm$0.01  &  \\[1mm]
    & $2^{13}$ &                                      & 8.05\tiny$\pm$1.11  &  \\[1mm]
    & $2^{15}$ &                                      & 86.17\tiny$\pm$23.33 &  \\[1mm]
    & $2^{17}$ &                                      & oot                &  \\
\midrule
\multirow[t]{10}{*}{P4} 
    & $2^{9}$  & {\tiny\texttt{EF(a = 1 \& AG(r $\neq$ 1))}} & 0.05\tiny$\pm$0.00  & oot  \\[1mm]
    & $2^{11}$ &                                      & 0.45\tiny$\pm$0.01  &  \\[1mm]
    & $2^{13}$ &                                      & 5.93\tiny$\pm$0.04  & \\[1mm]
    & $2^{15}$ &                                      & 101.30\tiny$\pm$4.95 &  \\[1mm]
    & $2^{17}$ &                                      & oot                &  \\[1mm]
    & $2^{9}$  & {\tiny\texttt{!(EF(a = 1 \& AG(r $\neq$ 1)))} } & 0.05\tiny$\pm$0.00  &  \\[1mm]
    & $2^{11}$ &                                      & 0.40\tiny$\pm$0.02  &  \\[1mm]
    & $2^{13}$ &                                      & 5.91\tiny$\pm$0.05  &  \\[1mm]
    & $2^{15}$ &                                      & 98.77\tiny$\pm$4.72  &  \\[1mm]
    & $2^{17}$ &                                      & oot                &  \\
\midrule
\multirow[t]{10}{*}{P5} 
    & $2^{9}$  & {\tiny\texttt{AG(s $\neq$ 1 | AF(u = 1))}}  & 0.08\tiny$\pm$0.00  & 1.78\tiny$\pm$0.35  \\[1mm]
    & $2^{11}$ &                                      & 1.07\tiny$\pm$0.01  &  \\[1mm]
    & $2^{13}$ & {\tiny\texttt{EG(a = 1 \& EF(r $\neq$ 1))}}  & 6.23\tiny$\pm$0.09  & \\[1mm]
    & $2^{15}$ &                                      & 109.82\tiny$\pm$9.69 &  \\[1mm]
    & $2^{17}$ &                                      & oot                &  \\[1mm]
    & $2^{9}$  & {\tiny\texttt{!(AG(s $\neq$ 1 | AF(u = 1)))} } & 0.09\tiny$\pm$0.00  &  \\[1mm]
    & $2^{11}$ &                                      & 1.12\tiny$\pm$0.02  &  \\[1mm]
    & $2^{13}$ & {\tiny\texttt{!(EG(a = 1 \& EF(r $\neq$ 1)))} } & 6.02\tiny$\pm$0.08  &  \\[1mm]
    & $2^{15}$ &                                      & 102.62\tiny$\pm$4.09 &  \\[1mm]
    & $2^{17}$ &                                      & oot                &  \\
\midrule
\multirow[t]{10}{*}{P6} 
    & $2^{9}$  & {\tiny\texttt{EF(s = 1 | EG(u $\neq$ 1))}}  & 0.08\tiny$\pm$0.01  &  0.36\tiny$\pm$0.01 \\[1mm]
    & $2^{11}$ &                                      & 1.30\tiny$\pm$0.02  &  \\[1mm]
    & $2^{13}$ & {\tiny\texttt{EF(a = 1 \& AG(r $\neq$ 1))}}  & 5.83\tiny$\pm$0.04  & \\[1mm]
    & $2^{15}$ &                                      & 94.51\tiny$\pm$6.33  &  \\[1mm]
    & $2^{17}$ &                                      & oot                &  \\[1mm]
    & $2^{9}$  & {\tiny\texttt{!(EF(s = 1 | EG(u $\neq$ 1)))} } & 0.08\tiny$\pm$0.00  &  \\[1mm]
    & $2^{11}$ &                                      & 1.30\tiny$\pm$0.03  &  \\[1mm]
    & $2^{13}$ & {\tiny\texttt{!(EF(a = 1 \& AG(r $\neq$ 1)))} } & 5.73\tiny$\pm$0.07  &  \\[1mm]
    & $2^{15}$ &                                      & 92.48\tiny$\pm$5.61  &  \\[1mm]
    & $2^{17}$ &                                      & oot                &  \\
\midrule
\multirow[t]{10}{*}{P7} 
    & $2^{9}$  & {\tiny\texttt{AG(s $\neq$ 1 | EF(u = 1))}}   & 0.08\tiny$\pm$0.00  & 0.36\tiny$\pm$0.00 \\[1mm]
    & $2^{11}$ &                                      & 1.30\tiny$\pm$0.01  &  \\[1mm]
    & $2^{13}$ & {\tiny\texttt{AG(a $\neq$ 1 \& EF(r = 1))}}  & 4.62\tiny$\pm$0.10  &  \\[1mm]
    & $2^{15}$ &                                      & 75.24\tiny$\pm$3.40  &  \\[1mm]
    & $2^{17}$ &                                      & oot                &  \\[1mm]
    & $2^{9}$  & {\tiny\texttt{!(AG(s $\neq$ 1 | EF(u = 1)))} } & 0.08\tiny$\pm$0.00  &  \\[1mm]
    & $2^{11}$ &                                      & 1.33\tiny$\pm$0.03  &  \\[1mm]
    & $2^{13}$ & {\tiny\texttt{!(AG(a $\neq$ 1 \& EF(r = 1)))} } & 8.39\tiny$\pm$0.13  &  \\[1mm]
    & $2^{15}$ &                                      & 85.20\tiny$\pm$9.46  &  \\[1mm]
    & $2^{17}$ &                                      & oot                &  \\
\midrule
\multirow[t]{10}{*}{P17} 
    & $2^{9}$  & {\tiny\texttt{AG(AF(x $\ge$ 1))}}          & 0.03\tiny$\pm$0.00  & 0.53\tiny$\pm$0.09 \\[1mm]
    & $2^{11}$ &                                      & 0.18\tiny$\pm$0.00  &  \\[1mm]
    & $2^{13}$ &                                      & 9.80\tiny$\pm$0.20  &  \\[1mm]
    & $2^{15}$ &                                      & 159.03\tiny$\pm$4.76 &  \\[1mm]
    & $2^{17}$ &                                      & oot                &  \\[1mm]
    & $2^{9}$  & {\tiny\texttt{!(AG(AF(x $\ge$ 1)))} }       & 0.03\tiny$\pm$0.00  &  \\[1mm]
    & $2^{11}$ &                                      & 0.19\tiny$\pm$0.00  &  \\[1mm]
    & $2^{13}$ &                                      & 12.12\tiny$\pm$0.44 &  \\[1mm]
    & $2^{15}$ &                                      & 168.02\tiny$\pm$13.92 &  \\[1mm]
    & $2^{17}$ &                                      & oot                &  \\
\midrule
\multirow[t]{10}{*}{P18} 
    & $2^{9}$  & {\tiny\texttt{EF(EG(x < 1))}}             & 0.03\tiny$\pm$0.00  & 0.42\tiny$\pm$0.01  \\[1mm]
    & $2^{11}$ &                                      & 0.19\tiny$\pm$0.01  &  \\[1mm]
    & $2^{13}$ &                                      & 9.87\tiny$\pm$0.27  & \\[1mm]
    & $2^{15}$ &                                      & 161.75\tiny$\pm$2.16 &  \\[1mm]
    & $2^{17}$ &                                      & oot                &  \\[1mm]
    & $2^{9}$  & {\tiny\texttt{!(EF(EG(x < 1)))} }          & 0.03\tiny$\pm$0.00  &  \\[1mm]
    & $2^{11}$ &                                      & 0.18\tiny$\pm$0.00  &  \\[1mm]
    & $2^{13}$ &                                      & 9.73\tiny$\pm$0.34  &  \\[1mm]
    & $2^{15}$ &                                      & 160.55\tiny$\pm$2.19 &  \\[1mm]
    & $2^{17}$ &                                      & oot                &  \\
\midrule
\multirow[t]{10}{*}{P19} 
    & $2^{9}$  & {\tiny\texttt{AG(EF(x $\ge$ 1))}}          & 0.03\tiny$\pm$0.00  & 0.54\tiny$\pm$0.12 \\[1mm]
    & $2^{11}$ &                                      & 0.20\tiny$\pm$0.00  &  \\[1mm]
    & $2^{13}$ &                                      & 11.84\tiny$\pm$0.32 &  \\[1mm]
    & $2^{15}$ &                                      & 194.40\tiny$\pm$3.21 &  \\[1mm]
    & $2^{17}$ &                                      & oot                &  \\[1mm]
    & $2^{9}$  & {\tiny\texttt{!(AG(EF(x $\ge$ 1)))} }       & 0.03\tiny$\pm$0.00  &  \\[1mm]
    & $2^{11}$ &                                      & 0.20\tiny$\pm$0.00  &  \\[1mm]
    & $2^{13}$ &                                      & 11.81\tiny$\pm$0.35 &  \\[1mm]
    & $2^{15}$ &                                      & 198.88\tiny$\pm$3.56 &  \\[1mm]
    & $2^{17}$ &                                      & oot                &  \\
\midrule
\multirow[t]{10}{*}{P20} 
    & $2^{9}$  & {\tiny\texttt{EF(AG(x < 1))}}             & 0.03\tiny$\pm$0.00  & 0.67\tiny$\pm$0.10 \\[1mm]
    & $2^{11}$ &                                      & 0.20\tiny$\pm$0.01  &  \\[1mm]
    & $2^{13}$ &                                      & 11.84\tiny$\pm$0.19 &  \\[1mm]
    & $2^{15}$ &                                      & 190.75\tiny$\pm$2.87 &  \\[1mm]
    & $2^{17}$ &                                      & oot                &  \\[1mm]
    & $2^{9}$  & {\tiny\texttt{!(EF(AG(x < 1)))} }          & 0.03\tiny$\pm$0.00  &  \\[1mm]
    & $2^{11}$ &                                      & 0.20\tiny$\pm$0.00  &  \\[1mm]
    & $2^{13}$ &                                      & 13.64\tiny$\pm$0.67 &  \\[1mm]
    & $2^{15}$ &                                      & 192.30\tiny$\pm$4.08 &  \\[1mm]
    & $2^{17}$ &                                      & oot                &  \\
\midrule
\multirow[t]{10}{*}{P21} 
    & $2^{9}$  & {\tiny\texttt{AG(AF(w = 1))}}             & 0.02\tiny$\pm$0.00  & 0.37\tiny$\pm$0.00 \\[1mm]
    & $2^{11}$ &                                      & 0.09\tiny$\pm$0.00  &  \\[1mm]
    & $2^{13}$ &                                      & 4.98\tiny$\pm$0.34  &  \\[1mm]
    & $2^{15}$ &                                      & 28.79\tiny$\pm$22.88 &  \\[1mm]
    & $2^{17}$ &                                      & oot                &  \\[1mm]
    & $2^{9}$  & {\tiny\texttt{!(AG(AF(w = 1)))} }          & 0.02\tiny$\pm$0.00  &  \\[1mm]
    & $2^{11}$ &                                      & 0.10\tiny$\pm$0.00  &  \\[1mm]
    & $2^{13}$ &                                      & 1.03\tiny$\pm$0.00  &  \\[1mm]
    & $2^{15}$ &                                      & 16.23\tiny$\pm$0.17 &  \\[1mm]
    & $2^{17}$ &                                      & 258.78\tiny$\pm$1.51  &  \\
\midrule
\multirow[t]{10}{*}{P22} 
    & $2^{9}$  & {\tiny\texttt{EF(EG(w $\neq$ 1))}}          & 0.02\tiny$\pm$0.00  & 0.68\tiny$\pm$0.06 \\[1mm]
    & $2^{11}$ &                                      & 0.10\tiny$\pm$0.00  &  \\[1mm]
    & $2^{13}$ &                                      & 1.02\tiny$\pm$0.00  &  \\[1mm]
    & $2^{15}$ &                                      & 19.26\tiny$\pm$6.30  &  \\[1mm]
    & $2^{17}$ &                                      & oot                &  \\[1mm]
    & $2^{9}$  & {\tiny\texttt{!(EF(EG(w $\neq$ 1)))} }       & 0.02\tiny$\pm$0.00  &  \\[1mm]
    & $2^{11}$ &                                      & 0.10\tiny$\pm$0.00  &  \\[1mm]
    & $2^{13}$ &                                      & 1.09\tiny$\pm$0.07  &  \\[1mm]
    & $2^{15}$ &                                      & 19.86\tiny$\pm$6.23  &  \\[1mm]
    & $2^{17}$ &                                      & oot                &  \\
\midrule
\multirow[t]{10}{*}{P23} 
    & $2^{9}$  & {\tiny\texttt{AG(EF(w = 1))}}           & 0.08\tiny$\pm$0.00  & 1.42\tiny$\pm$0.20 \\[1mm]
    & $2^{11}$ &                                      & 1.02\tiny$\pm$0.02  &  \\[1mm]
    & $2^{13}$ &                                      & 19.09\tiny$\pm$0.24  &  \\[1mm]
    & $2^{15}$ &                                      & 404.94\tiny$\pm$2.68 &  \\[1mm]
    & $2^{17}$ &                                      & oot                &  \\[1mm]
    & $2^{9}$  & {\tiny\texttt{!(AG(EF(w = 1)))} }        & 0.08\tiny$\pm$0.00  &  \\[1mm]
    & $2^{11}$ &                                      & 1.04\tiny$\pm$0.02  &  \\[1mm]
    & $2^{13}$ &                                      & 19.45\tiny$\pm$0.27  &  \\[1mm]
    & $2^{15}$ &                                      & 428.07\tiny$\pm$2.42 &  \\[1mm]
    & $2^{17}$ &                                      & oot                &  \\
\midrule
\multirow[t]{10}{*}{P24} 
    & $2^{9}$  & {\tiny\texttt{EF(AG(w $\neq$ 1))}}        & 0.04\tiny$\pm$0.00  & 0.12\tiny$\pm$0.00 \\[1mm]
    & $2^{11}$ &                                      & 0.27\tiny$\pm$0.01  &  \\[1mm]
    & $2^{13}$ &                                      & 6.45\tiny$\pm$0.07  &  \\[1mm]
    & $2^{15}$ &                                      & 71.41\tiny$\pm$20.81 &  \\[1mm]
    & $2^{17}$ &                                      & oot                &  \\[1mm]
    & $2^{9}$  & {\tiny\texttt{!(EF(AG(w $\neq$ 1)))} }     & 0.04\tiny$\pm$0.00  &  \\[1mm]
    & $2^{11}$ &                                      & 0.26\tiny$\pm$0.00  &  \\[1mm]
    & $2^{13}$ &                                      & 6.60\tiny$\pm$0.14  &  \\[1mm]
    & $2^{15}$ &                                      & 93.10\tiny$\pm$23.08 &  \\[1mm]
    & $2^{17}$ &                                      & oot                &  \\
\midrule
\multirow[t]{10}{*}{P25} 
    & $2^{9}$  & {\tiny\texttt{(c > 5) -> (AF(r > 5))}}     & 0.06\tiny$\pm$0.01  & 0.31\tiny$\pm$0.00 \\[1mm]
    & $2^{11}$ &                                      & 1.03\tiny$\pm$0.01  &  \\[1mm]
    & $2^{13}$ &                                      & 28.75\tiny$\pm$0.27  &  \\[1mm]
    & $2^{15}$ &                                      & 272.40\tiny$\pm$60.53 &  \\[1mm]
    & $2^{17}$ &                                      & oot                &  \\[1mm]
    & $2^{9}$  & {\tiny\texttt{!((c > 5) -> (AF(r > 5)))} }  & 0.06\tiny$\pm$0.01  &  \\[1mm]
    & $2^{11}$ &                                      & 1.05\tiny$\pm$0.02  &  \\[1mm]
    & $2^{13}$ &                                      & 28.62\tiny$\pm$0.43  &  \\[1mm]
    & $2^{15}$ &                                      & 262.69\tiny$\pm$60.53 &  \\[1mm]
    & $2^{17}$ &                                      & oot                &  \\
\bottomrule
\label{tab:ctlfin}
\end{longtable}
}
\end{document}